%% file: paper.tex
\documentclass[11pt,a4paper]{article}
\usepackage{amssymb,amsmath,amsthm}
\usepackage{graphicx}
\usepackage{hyperref}
\usepackage{subcaption}
\usepackage{booktabs,tabu}

\renewcommand{\leq}{\leqslant}
\renewcommand{\geq}{\geqslant}
\newcommand{\src}{\mathit{source}}
\newcommand{\sink}{\mathit{sink}}
\newcommand{\dm}{\mathrm{D}}
\newcommand{\im}{\mathrm{Im}}
\newcommand{\poly}{\text{poly}}
\newcommand{\eps}{\varepsilon}
\newcommand{\val}{\mathit{value}}
\newcommand{\overbar}[1]{\mkern 1.5mu\overline{\mkern-3.5mu#1\mkern-1.5mu}}

\newenvironment{emromani} {%
\begin{enumerate}}
{\end{enumerate}}

\newtheorem{theorem}{Theorem}[section]

\newtheorem{claim}[theorem]{\bfseries{Claim}}
\newtheorem{lemma}[theorem]{\bfseries{Lemma}}

\newbox\ProofSym
\setbox\ProofSym=\hbox{%
\unitlength=0.18ex%
\begin{picture}(10,10)
\put(0,0){\framebox(9,9){}}
\put(0,3){\framebox(6,6){}}
\end{picture}}

\begin{document}

\input{src/title}
\input{src/introduction}
\input{src/preliminary}
\input{src/algorithm}
\input{src/analysis}
\input{src/conclusion}

\bibliography{paper}
\end{document}

%% file: src/title.tex
\title{
	Restricted Max-Min Fair Allocation\footnote{
		A preliminary version appears in the Proceedings of the 45th International Colloquium on Automata, Languages, and Programming(ICALP), 37:1-37:13.
	}
	\thanks{ 
		Research supported by the Research Grants Council, 
		Hong Kong, China (project no.~16201116).
	}
}

\author{
	Siu-Wing Cheng\footnote{Department of Computer Science and Engineering, HKUST. \{scheng, ymaoad\}@cse.ust.hk} 
	\and 
	Yuchen Mao\footnotemark[3]
}

\date{\today}

\maketitle

\begin{abstract}
	The restricted max-min fair allocation problem seeks an allocation of resources to players that maximizes the minimum total value obtained by any player.  It is NP-hard to approximate the problem to a ratio less than 2.  Comparing the current best algorithm for estimating the optimal value with the current best for constructing an allocation, there is quite a gap between the ratios that can be achieved in polynomial time: roughly 4 for estimation and roughly $6 + 2\sqrt{10}$ for construction.  We propose an algorithm that constructs an allocation with value within a factor of $6 + \delta$ from the optimum for any constant $\delta > 0$.  The running time is polynomial in the input size for any constant $\delta$ chosen.
 \end{abstract}

%% file: src/introduction.tex
\section{Introduction}
\label{sec:intro}
\paragraph*{Background.}
Let $P$ be a set of $m$ players.  Let $R$ be a set of $n$ indivisible resources.  Resource $r \in R$ is worth a non-negative integer value $v_{pr}$ for player $p \in P$.  An allocation is a partition of $R$ into disjoint subsets $\{C_p : p \in P\}$ so that player $p$ is assigned the resources in $C_p$.  The \emph{max-min fair allocation} problem is to distribute resources to players so that the minimum total value of resources received by any player is maximized.  We define the \emph{value of an allocation} to be $\min_{p \in P} \sum_{r \in C_p} v_{pr}$.  Equivalently, we want to find an allocation with maximum value.

Bez\'{a}kov\'{a} and Dani~\cite{BD05} attacked the problem using the techniques of Lenstra et al.~\cite{LST90} for the min-max version: the problem of scheduling on unrelated machine to minimize makespan.  Bez\'{a}kov\'{a} and Dani proved that no polynomial-time algorithm can give an approximation ratio less than 2 unless P $=$ NP.  However, the assignment LP used in~\cite{LST90} cannot be rounded to give an approximation for the max-min allocation problem because the integrality gap is unbounded.  Later, Bansal and Sviridenko~\cite{BS06} proposed a stronger LP relaxation, the configuration LP, for the max-min allocation problem.  They showed that although the configuration LP has exponentially many constraints, it can be solved to any desired accuracy in polynomial time.  They also showed that there is an integrality gap of $\Omega(\sqrt{m})$.   Asadpour and Saberi~\cite{AS07} developed a polynomial-time rounding scheme for the configuration LP that gives an approximation ratio of $O(\sqrt{m}\log^3 m)$.  Saha and Srinivasan~\cite{SS10} improved it to $O(\sqrt{m\log m})$.  Chakrabarty, Chuzhoy, and Khanna~\cite{CCK09} showed that an $O(n^\delta\log n)$-approximate allocation can be computed in $n^{O(1/\delta)}$ time for any $\delta \geq \frac{9\log\log n}{\log n}$.

In this paper, we focus on the \emph{restricted} max-min fair allocation problem.  In the restricted case, each resource is desired by some subset of players, and has the same value $v_r$ for those who desire it and value 0 for the rest.  Even in this case, no approximation ratio better than 2 can be obtained unless P $=$ NP~\cite{BD05}.  Bansal and Sviridenko~\cite{BS06} proposed a polynomial-time  $O\bigl(\frac{\log\log m}{\log\log\log m}\bigr)$-approximation algorithm which is based on rounding the configuration LP.  Feige~\cite{F08} proved that the integrality gap of the configuration LP is bounded by a constant (large and unspecified).  His proof was made constructive by Haeupler et al.~\cite{HSS11}, and hence, a constant approximation can be found in polynomial time. Adapting Haxell's techniques for hypergraph bipartite matching~\cite{H95}, Asadpour et al.~\cite{AFS12} proved that the integrality gap of the configuration LP is at most 4.  Therefore, by solving the configuration LP approximately, one can estimate the optimal solution value within a factor of $4 + \delta$ in polynomial time for any constant $\delta > 0$. However, it is not known how to construct a $(4+\delta)$-approximate allocation in polynomial time.  Inspired by the ideas in~\cite{AFS12} and~\cite{H95}, Annamalai et al.~\cite{AKS17} developed a purely combinatorial algorithm that avoids solving the configuration LP.  It runs in polynomial time and guarantees an approximation ratio $6 + 2\sqrt{10} + \delta$ for any constant $\delta > 0$.  Nevertheless, the analysis still relies on the configuration LP.  There is quite a gap between the current best estimation ratio\footnote{It was recently improved to $3 + \frac{5}{6}$ independently in \cite{CM18-2,JR18}} $4+\delta$ and the current best approximation ratio $6+2\sqrt{10}+\delta \approx 12.325 +\delta$. This is an interesting status that few problems have.

If one constrains the \emph{restricted case} further by requiring $v_r\in\{1,\eps\}$ for some fixed constant $\eps \in (0,1)$, then it becomes the \emph{$(1,\eps)$-restricted case}.  Golovin proposed an $O(\sqrt{n})$-approximation algorithm for this case~\cite{G05}.  Chan et al.~\cite{CTW16} showed that it is still NP-hard to obtain an approximation ratio less than 2
and that the algorithm of Annamalai et al.~\cite{AKS17} achieves an approximation ratio of $9$ in this case.  The analysis in~\cite{CTW16} does not rely on the configuration LP.

\paragraph*{Our contributions.}
We propose an algorithm for the restricted max-min fair allocation problem that achieves an approximation ratio of $6 + \delta$ for any constant $\delta > 0$.  It runs in polynomial time for any constant $\delta$ chosen.  Our algorithm uses the same framework of Annamalai et al.~\cite{AKS17}: we maintain a stack of layers to record the relation between players and resources, and use lazy update and a greedy strategy to achieve a polynomial running time.  

Let $\tau^*$ be the optimal solution value.  Let $\lambda > 2$ be the target approximation ratio.  To obtain a $\lambda$-approximate solution, the value of resources a player need is $\tau^*/\lambda$. Our first contribution is a greedy strategy that is much more aggressive than that of Annamalai et al.~\cite{AKS17}.  Their greedy strategy considers a player greedy if that player claims at least $\tau^*/2$ worth of resources, which is more than needed.  In contrast, we consider a player greedy if it claims (nearly) the largest total value among all the candidates.  When building the stack, as in~\cite{AKS17}, we add greedy players and the resources claimed by them to the stack.  Intuitively, our more aggressive greedy strategy leads to a faster growth of the stack, and hence a significantly smaller approximation ratio can be achieved.

Our aggressive strategy brings challenge to the analysis that previous approaches~\cite{AKS17,CTW16} cannot cope with.  Our second contribution is a new analysis tool: an injection that maps a lot of players in the stack to their \emph{competing players} who can access resources of large total value.  Since players added to the stack must be greedy, they claim more than their competing players.  Therefore, such an injection allows us to conclude that players in the stack claim large worth of resources.  By incorporating competing players into the analysis framework of Chan et al.~\cite{CTW16}, we improve the approximation ratio to $6+\delta$.  Our analysis does not rely on the configuration LP, and it is purely combinatorial.

%% file: src/preliminary.tex
\section{Preliminaries}
\label{sec:pre}

Let $\tau^*$ be the optimal solution value.  Let $\lambda > 2$ denote our target approximation ratio.  Given any value $\tau \leq \tau^*$, our algorithm returns an allocation of value $\tau/\lambda$ in polynomial time.  We will show how to combine this algorithm with binary search to obtain an allocation of value at least $\tau^*/\lambda$ in the end.  We assume that $\tau$ is no more than $\tau^*$ in the rest of this section.

\subsection{Fat edges, thin edges and partial allocations}  
A resource $r$ is \emph{fat} if $v_r \geq \tau/\lambda$, and \emph{thin} otherwise.  For a set $B$ of thin resources, we define $\val(B) = \sum_{r\in B}v_r$.  For any player $p$, and any fat resource $r_f$ that is desired by $p$, $(p, r_f)$ is a \emph{fat edge}.  For any player $p$, and any set $B$ of thin resources, $(p, B)$ is a \emph{thin edge} if $p$ desires all the resources in $B$ and $\val(B) \geq \tau/\lambda$.  For a thin edge $e = (p, B)$, we say player $p$ and the resources in $B$ are covered by $e$, and define $\val(e) = \val(B) = \sum_{r\in B}v_r$.  We use uppercase calligraphic letters to denote sets of thin edges.  Given a set $\cal S$ of thin edges, we say $\cal S$ covers a player or a thin resource if some edge in $\cal S$ covers that player or resource, and define $\val({\cal S})$ to be the total value of the thin resources covered by ${\cal S}$.  That is, $\val({\cal S}) = \val(\bigcup_{(p, B)\in {\cal S} }B)$.  

Since our target approximation ratio is $\lambda$, a player will be \emph{satisfied} if it receives either a single fat resource it desires, or at least $\tau/\lambda$ worth of thin resources it desires.  Hence, it suffices to consider allocations that consist of two parts, one being a set of fat edges and the other being a set of thin edges.

Let $G$ be the bipartite graph formed by all the players, all the fat resources, and all the fat edges.  We will start with an arbitrary maximum matching $M$ of $G$ (which is a set of fat edges) and an empty set $\cal E$ of thin edges,  and iteratively update and grow $M$ and $\cal E$ into an allocation that satisfies all the players.  We call the intermediate solutions \emph{partial allocations} and formally define them as follows.  

A partial allocation consists of a maximum matching $M$ of $G$ and a subset $\cal E$ of thin edges such that (i) no two edges in $M$ and $\cal E$ satisfy (i.e., cover) the same player, (ii) no two edges in $\cal E$ share any resource, (iii) every edge $(p,B) \in {\cal E}$ is minimal in the sense that every proper subset $B' \subset B$ has value less than $\tau/\lambda$.  

In Section~\ref{sec:alg}, we present an algorithm which, given a partial allocation and an unsatisfied player $p_0$, computes a new partial allocation that satisfies $p_0$ and all the players that used to be satisfied.  Repeatedly invoking this algorithm returns an allocation that satisfies all the players.

\subsection{A problem of finding node-disjoint paths}  

We define a family of networks and a problem of finding node-disjoint paths in these networks.  These networks and the node-disjoint paths problem are used heavily in our algorithm and analysis.

\subsubsection{The problem} 
\label{subsubsec:path-problem}

\begin{figure}
	\centering
	\begin{subfigure}[b]{0.2\linewidth}
		\centering
		\includegraphics[scale = 0.7]{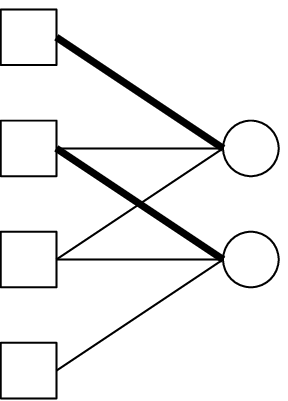}
		\caption{$M$}
	\end{subfigure}
	\begin{subfigure}[b]{0.2\linewidth}
		\centering
		\includegraphics[scale = 0.7]{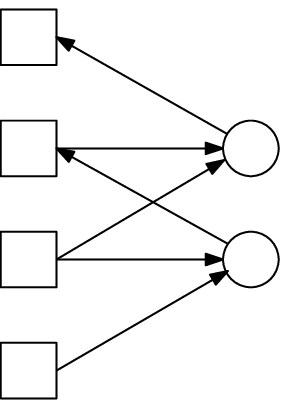}
		\caption{$G_M$}
	\end{subfigure}
	\begin{subfigure}[b]{0.22\linewidth}
		\centering
		\includegraphics[scale = 0.7]{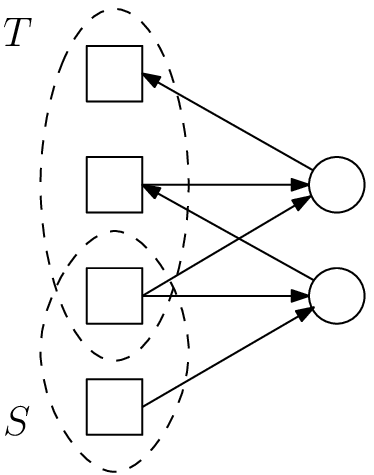}
		\caption{$G_M(S,T)$}
	\end{subfigure}
	\begin{subfigure}[b]{0.28\linewidth}
		\centering
		\includegraphics[scale = 0.7]{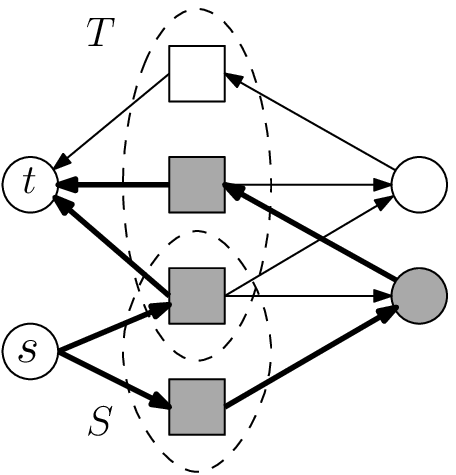}
		\caption{$s$-$t$ flow network}
	\end{subfigure}
	\caption{A reduction from the node-disjoint path problem to max flow problem.  The squares represent players and the circles represent fat resources.}
	\label{fg:s-t-flow}
\end{figure}

Recall that $G$ is a bipartite graph formed by all the players, all the fat resources, and all the fat edges. With respect to any maximum matching $M$ of $G$, we define $G_M$ to be a directed bipartite graph obtained from $G$ by orienting edges of $G$ from $r$ to $p$ if the edge $(p,r)$ is in $M$, and from $p$ to $r$ if $(p,r)$ is not in $M$. See Figure~\ref{fg:s-t-flow}(a) and (b) for an example.

We use $P_M$ and $\overbar{P}_M$ to denote the subsets of players matched and unmatched in $M$, respectively.  Given $S\subseteq \overbar{P}_M$ and $T\subseteq P$, we use $G_M(S,T)$ to denote the problem of finding the maximum number of node-disjoint paths from $S$ to $T$ in $G_M$.  This problem will arise in this paper for different choices of $S$ and $T$.  A feasible solution of $G_M(S,T)$ is just any set of node-disjoint paths from $S$ to $T$ in $G_M$.  An optimal solution maximizes the number of such paths.  Let $f_M(S,T)$ denote the size of an optimal solution of $G_M(S,T)$.  In the cases that $S \cap T \not= \emptyset$, a feasible solution may contain a path from a player $p \in S \cap T$ to itself, i.e., a path with no edge.   We call such a path a \emph{trivial path}.  Any path with at least one edge is \emph{non-trivial}.

Let $\Pi$ be any feasible solution of $G_M(S,T)$. The paths in $\Pi$ originate from a subset of $S$, which we call the \emph{sources}, and terminate at a subset of $T$, which we call the \emph{sinks}.  We denote the sets of sources and sinks by $\src(\Pi)$ and $\sink(\Pi)$, respectively.  A trivial path has only one node which is both its source and sink. From now on, we use $\Pi^+$ to denote the subset of non-trivial paths in $\Pi$.

\subsubsection{Solving the problem} 
\label{subsubsec:find-paths}

An optimal solution of $G_M(S,T)$ can be found by solving a maximum $s$-$t$ flow problem. Let $H$ be the $s$-$t$ flow network obtained from $G_M$ by adding a super source $s$ and directed edges from $s$ to all vertices in $S$,  adding a super sink $t$ and directed edges from all vertices in $T$ to $t$, and setting the capacities of all edges to 1.  It suffices to find an integral maximum flow in $H$.  The paths in $G_M$ used by this maximum flow is an optimal solution of $G_M(S,T)$.  Node-disjointness is guaranteed because, in $H$, every player has its in-degree at most one and every resource has its out-degree at most one.  

Figure~\ref{fg:s-t-flow} gives an example.  The squares represent players.  The circles represent fat resources.  In (a), the bold undirected edges form the maximum matching $M$.  The two lower square nodes are unmatched and they form $\overbar{P}_M$.  The two upper nodes are matched and they form $P_M$.  An optimal solution of $G_M(S,T)$ can be computed by finding an integral maximum $s$-$t$ flow in the network in (d).  The shaded nodes and bold edges in (d) form a maximum $s$-$t$ flow.  If you ignore $s$, $t$, and the edges incident to them, the remaining shaded nodes and bold edges form an optimal solution of $G_M(S,T)$, which contains one trivial path and one non-trivial path.

\subsubsection{Non-trivial paths and the \texorpdfstring{$\mathbf{\oplus}$}{XOR} operator} 
\label{subsubsec:alternating-path}
Let $\pi$ be a non-trivial path from $\overbar{P}_M$ to $P$ in $G_M$.  If we ignore the directions of edges in $\pi$, then $\pi$ is called an \emph{alternating path} in the matching literature~\cite{PL86}: the first edge of $\pi$ does not belong to $M$, every other edge of $\pi$ belongs to $M$, and $\pi$ has an even number of edges.  We use $M \oplus \pi$ to denote the result of flipping $\pi$, i.e., removing the edges in $\pi \cap M$ from the matching and adding the edges in $\pi \setminus M$ to the matching.  $M \oplus \pi$ is a maximum matching of $G$.  Moreover, $\src(\pi)$ is unmatched in $M$ but it becomes matched in $M \oplus \pi$, and $\sink(\pi)$ is matched in $M$ but it becomes unmatched in $M \oplus \pi$.

We can extend the above operation to any set $\Pi^+$ of node-disjoint non-trivial paths from $\overbar{P}_M$ to $P$ in $G_M$.  $\Pi^+$ can be regarded as a set of edges.  We can form $M \oplus \Pi^+$ as in the previous paragraph, i.e., ignore the directions of edges in $\Pi^+$, remove the edges in $\Pi^+ \cap M$ from the matching, and add the edges in $\Pi^+ \setminus M$ to the matching.  $M \oplus \Pi^+$ is a maximum matching of $G$. Players in $\src(\Pi^+)$ are unmatched in $M$ but they become matched in $M \oplus \Pi^+$, and players in $\sink(\Pi^+)$ are matched in $M$ but they become unmatched in $M \oplus \Pi^+$.

\subsubsection{Feasible solutions of \texorpdfstring{$G_M(S,T)$}{the path problem}} 

The preliminary background above are sufficient for understanding how our algorithm works. However, in order to carry out a rigorous analysis, we have to delve into the feasible solutions of $G_M(S, T)$.

\begin{figure}
	\centering
	\begin{subfigure}[b]{0.2\linewidth}
		\centering
		\includegraphics[scale = 0.7]{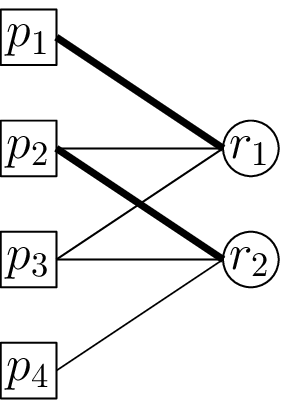}
		\caption{$M$}
	\end{subfigure}
	\begin{subfigure}[b]{0.22\linewidth}
		\centering
		\includegraphics[scale = 0.7]{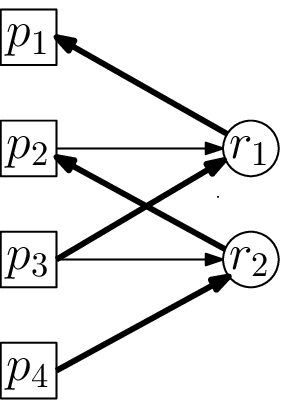}
		\caption{$G_M$ and $\Pi^+$}
	\end{subfigure}
	\begin{subfigure}[b]{0.2\linewidth}
		\centering
		\includegraphics[scale = 0.7]{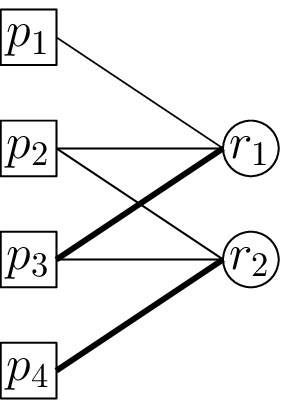}
		\caption{$M\oplus \Pi^+$}
	\end{subfigure}
	\begin{subfigure}[b]{0.2\linewidth}
		\centering
		\includegraphics[scale = 0.7]{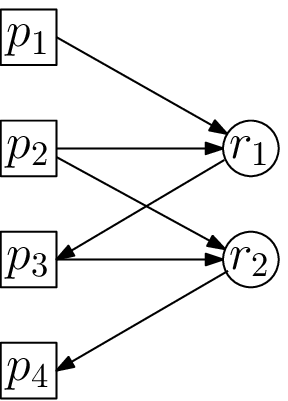}
		\caption{$G_{M\oplus \Pi^+}$}
	\end{subfigure}
	\caption{An illustration of the $\oplus$ operation and the relation between $G_M$ and $G_{M\oplus \Pi^+}$.}
	\label{fg:path-oplus}
\end{figure}

First let's discuss the $\oplus$ operation further.  Let $\Pi^+$ be any set of node-disjoint non-trivial paths from $\overbar{P}_M$ to $P$ in $G_M$.  $M\oplus \Pi^+$ is a maximum matching of $G$.  Now consider $G_{M\oplus \Pi^+}$, the directed bipartite graph defined for the maximum matching $M\oplus \Pi^+$ as in section~\ref{subsubsec:path-problem}.  We claim that $G_{M\oplus \Pi^+}$ can be interpreted as a graph obtained from $G_M$ by reversing the edges used by $\Pi^+$: when edges in $\Pi^+ \cap M$ are removed from the matching and edges in $\Pi^+ \setminus M$ are added to the matching, their counterparts in $G_M$ are reversed. Figure~\ref{fg:path-oplus} gives an example.  In (a), the bold edges form a maximum matching $M$.  (b) shows $G_M$, and $\Pi^+$ consists of the two bold paths, one from $p_3$ to $p_1$ and the other from $p_4$ to $p_2$.   In (c), the bold edges form a maximum matching $M\oplus \Pi^+$ which is obtained from $M$ by flipping the edges in $\Pi^+$. $G_{M \oplus \Pi^+}$ is shown in (d). Comparing (b) and (d), it is easy to see that $G_{M\oplus \Pi^+}$ can be obtained from $G_M$ by reversing the edges in $\Pi^+$.

Now we are ready to establish a few properties for feasible solutions of $G_M(S,T)$ that will be used later in the analysis of our algorithm.  As we explained in section~\ref{subsubsec:find-paths}, computing an optimal solution of $G_M(S,T)$ can be reduced to computing an integral maximum $s$-$t$ flow.  Consequently, feasible solutions of $G_M(S,T)$ have some properties that are similar to those in the max-flow literature.  Claims~\ref{cl:path-opt},~\ref{cl:path-aug} and~\ref{cl:path-reroute} are very much like testing the optimality of a flow, augmenting a flow, and rerouting a flow, respectively.  Recall that for a set $\Pi$ of node-disjoint paths, $\Pi^+$ is the subset of non-trivial paths in $\Pi$.

\begin{claim}
	\label{cl:path-opt}
	Let $\Pi$ be a feasible solution of $G_M(S, T)$.   
	$\Pi$ is an optimal solution of $G_M(S, T)$ if and only if $G_{M\oplus \Pi^+}$ contains no path from $S \setminus \src(\Pi)$ to $T \setminus sink(\Pi)$.
\end{claim}
\begin{proof}
	Let $H$ be the $s$-$t$ flow network constructed from $G_M$ as in section~\ref{subsubsec:find-paths}.  One can extend $\Pi$ to a flow in $H$ with value $|\Pi|$ by pushing unit flows from $s$ to $\src(\Pi)$, along $\Pi$, and then from $\sink(\Pi)$ to $t$.  Denote this flow by $F$.  Let $H_F$ be the residual graph of $H$ with respect to $F$.  $H_F$ can be obtained from $H$ by reversing the edges used by $F$.  Recall that $G_{M\oplus \Pi^+}$ can be obtained from $G_M$ by reversing the edges in $\Pi^+$.  Hence, if you ignore the $s$, $t$, and the edges incident to them, the remaining of the residual graph $H_F$ is exactly $G_{M\oplus \Pi^+}$. If there is a path $\pi$ in $G_{M\oplus \Pi^+}$ from $S \setminus \src(\Pi)$ to $T \setminus sink(\Pi)$, then the concatenation $s \cdot \pi \cdot t$ is a path in the residual graph $H_F$.  This means that we can augment $F$ using $s \cdot \pi \cdot t$ to increase the flow value and obtain more node-disjoints paths from $S$ to $T$ in $G_M$. (The augmentation may produce some unit-flow cycle(s) in $H$, and such cycles can be simply ignored when extracting the node-disjoint paths in $G_M$ from $S$ to $T$.)   If such a path $\pi$ does not exist, then $F$ is a maximum flow which proves the optimality of $\Pi$.
\end{proof}

The proof of Claim~\ref{cl:path-opt} immediately implies Claim~\ref{cl:path-aug}.

\begin{claim}
	\label{cl:path-aug}
	Let $\Pi$ be a feasible solution of $G_M(S, T)$.  Suppose that $G_{M\oplus \Pi^+}$ contains a path $\pi$ from $S \setminus \src(\Pi)$ to $T \setminus \sink(\Pi)$.  We can use $\pi$ to augment $\Pi$ to a feasible solution $\Pi'$ of $G_M(S, T)$ such that $|\Pi'| = |\Pi| + 1$, the vertex set  of $\Pi'$ is a subset of the vertices in $\Pi\cup\{\pi\}$,  $\src(\Pi') = \src(\Pi) \cup \{\src(\pi)\}$, and $\sink(\Pi') = \sink(\Pi) \cup \{\sink(\pi)\}$.
\end{claim}

\begin{figure}
	\centering
	\begin{subfigure}[t]{0.18\linewidth}
		\centering
		\includegraphics[scale = 0.5]{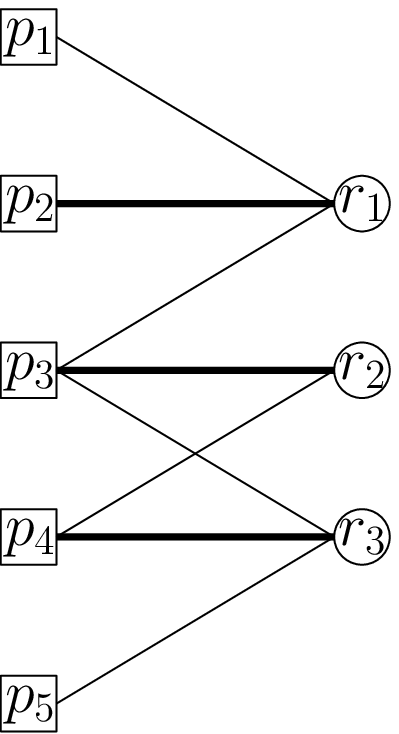}
		\caption{$M$}
	\end{subfigure}
	\begin{subfigure}[t]{0.25\linewidth}
		\centering
		\includegraphics[scale = 0.5]{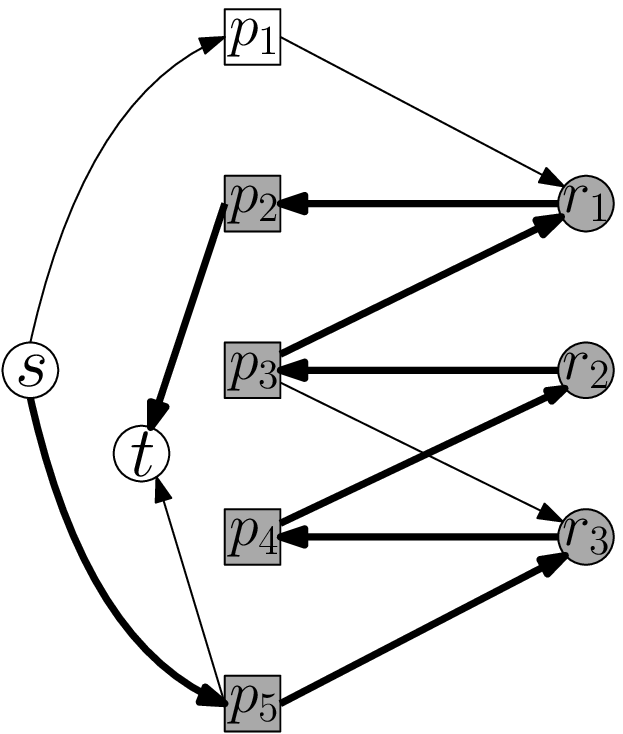}
		\caption{The $s$-$t$ flow network obtained from $G_M$}
	\end{subfigure}
	\begin{subfigure}[t]{0.25\linewidth}
		\centering
		\includegraphics[scale = 0.5]{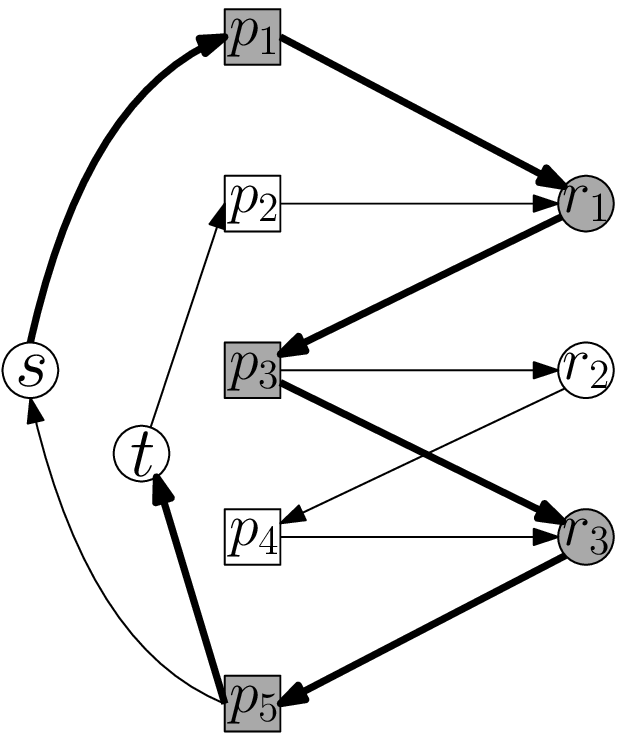}
		\caption{Residual graph}
	\end{subfigure}
	\begin{subfigure}[t]{0.25\linewidth}
		\centering
		\includegraphics[scale = 0.5]{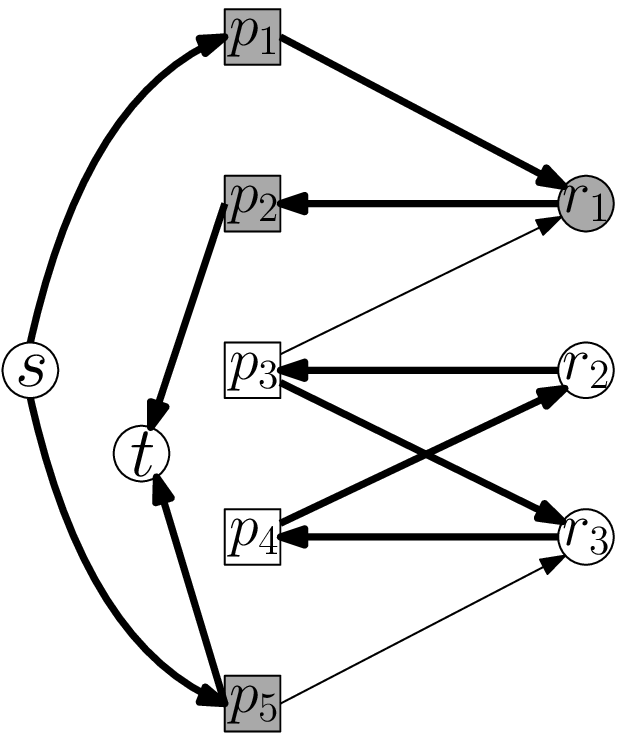}
		\caption{Augmented flow}
	\end{subfigure}
	\caption{An example of the proof of Claim~\ref{cl:path-opt} with $S = \{p_1, p_5\}$
	and $T = \{p_2, p_5\}$}
	\label{fg:path-aug}
\end{figure}

Figure~\ref{fg:path-aug} illustrates the proof of Claim~\ref{cl:path-opt}.  The maximum matching $M$ consists of the bold edges in (a).  In (b), the bold edges form an $s$-$t$ flow $F$.  The bold edges other than those incident to $s$ and $t$ form a feasible solution $\Pi$ of $G_M(S,T)$, which consist of a single path from $p_5$ to $p_2$. Note that $\Pi^+ = \Pi$ in this case.  The residual graph of the flow network in (b) with respect to $F$ is shown in (c). If we ignore $s$ and $t$ in (c), the subgraph is exactly $G_{M\oplus \Pi^+}$. The bold edges form an augmenting path $s \cdot \pi \cdot t$, where $\pi$ is a path in $G_{M\oplus \Pi^+}$. In (d), the bold edges form an $s$-$t$ flow $F'$ which is obtained from $F$ by augmenting along $s \cdot \pi \cdot t$.  $F'$ naturally induces a set $\Pi'$ of two node-disjoint paths from $S$ to $T$ in $G_M$: one trivial path from $p_5$ to itself and one non-trivial path from $p_1$ to $p_2$. The cycle $p_3r_3p_4r_2$ in $F'$ is ignored. One can check that $\Pi$ and $\Pi'$ satisfy Claim~\ref{cl:path-aug}.

\begin{claim}
	\label{cl:path-reroute}
	Let $\Pi$ be a feasible solution of $G_M(S, T)$.  Suppose that there is a non-trivial path $\pi$ in $G_{M\oplus \Pi^+}$ from $\sink(\Pi)$ to $T$.  Then it must be that $\sink(\pi) \notin \sink(\Pi)$, and we can use $\pi$ to convert $\Pi$ to another feasible solution $\Pi'$ of $G_M(S, T)$ such that $|\Pi'| = |\Pi|$, the vertex set of $\Pi'$ is a subset of the vertices in $\Pi\cup\{\pi\}$, $\src(\Pi') = \src(\Pi)$, and $\sink(\Pi') = \left(\sink(\Pi)\setminus \{\src(\pi)\}\right) \cup \{\sink(\pi)\}$.
\end{claim}
\begin{proof}
	Every node in $\sink(\Pi)$ is unmatched in $M\oplus \Pi^+$, and hence has zero in-degree in $G_{M\oplus \Pi^+}$.  Therefore, they cannot be the sink of a non-trivial path in $G_M$.  $\sink(\pi) \notin \sink(\Pi)$.  

	As in the proof of Claim~\ref{cl:path-opt}, let $H$ be the $s$-$t$ flow network constructed from $G_M$, let $F$ be the flow in $H$ corresponding to $\Pi$, and let $H_F$ be the residual graph of $H$ with respect to $F$.  Since $G_{M\oplus \Pi^+}$ is a subgraph of $H_F$, the path $\pi$ is also a path in $H_F$ from a player in $\sink(\Pi)$ to a player in $T$.  Since $\src(\pi) \in \sink(\Pi)$, there is an edge directed from $t$ to $\src(\pi)$ in $H_F$.  Since $\sink(\pi) \notin \sink(\Pi)$ and $\sink(\pi) \in T$, there is an edge directed from $\sink(\pi)$ to $t$ in $H_F$.  Therefore, $t \cdot \pi \cdot t$ is a cycle in $H_F$.  We update $F$ to another flow $F'$ by sending a unit flow along $t \cdot \pi \cdot t$.  After removing all cycle(s) of flows in $F'$ and removing all edges in $F'$ incident to $s$ and $t$, we obtain a set $\Pi'$ of node-disjoint paths from $S$ to $T$ in $G_M$.

	Since sending flow around a cycle does not change the total flow, the values of $F$ and $F'$ are equal, implying that $|\Pi| = |\Pi'|$.  By sending the unit flow around $t \cdot \pi \cdot t$, we do not update the flow on directed edges incident to $s$ in $H$.  Thus, every player who received a unit flow from $s$ before the update still receives a unit flow from $s$ afterwards, so $\src(\Pi') = \src(\Pi)$.  Since we push a flow from $t$ to $\src(\pi) \in \sink(\Pi)$, $\src(\pi)$ no longer sends a unit flow to $t$ in $H$, and is no longer a sink after the update.  As we push a flow from $\sink(\pi)\in T$ to $t$, $\sink(\pi)$ becomes a new sink.  All the other sinks are not affected.   We conclude that $\sink(\Pi') = \left(\sink(\Pi) \setminus \{\src(\pi)\} \right) \cup \{\sink(\pi)\}$.
\end{proof}

\begin{figure}
	\centering
	\begin{subfigure}[t]{0.18\linewidth}
		\centering
		\includegraphics[scale = 0.5]{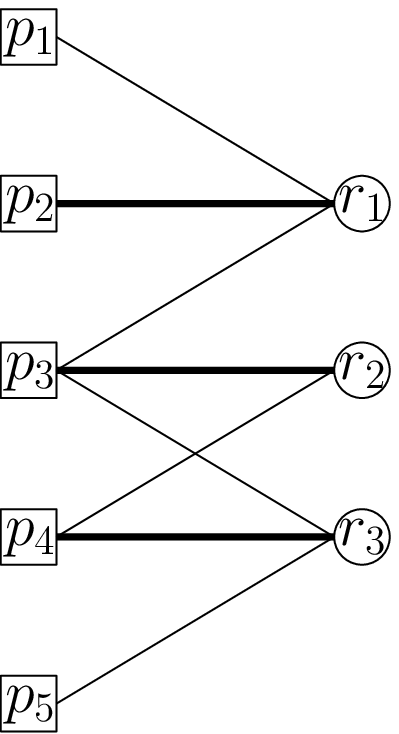}
		\caption{$M$}
	\end{subfigure}
	\begin{subfigure}[t]{0.25\linewidth}
		\centering
		\includegraphics[scale = 0.5]{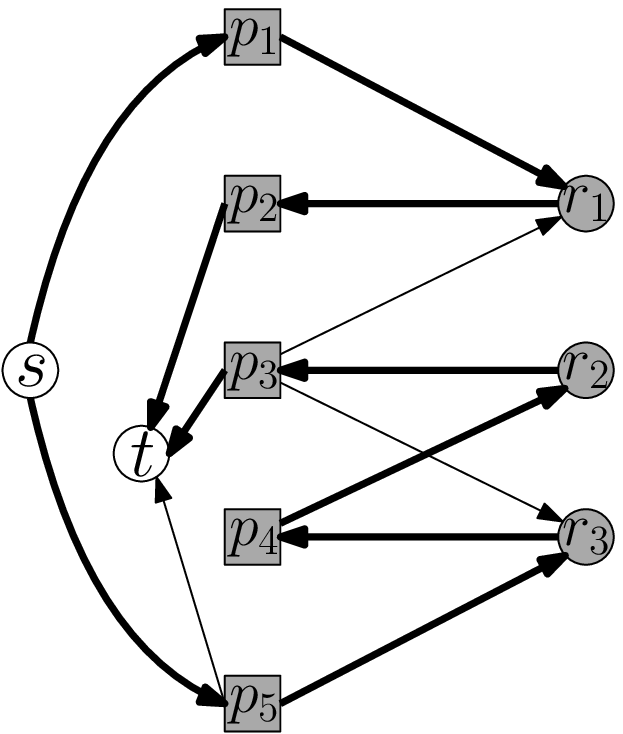}
		\caption{The $s$-$t$ flow network\\obtained from $G_M$}
	\end{subfigure}
	\begin{subfigure}[t]{0.25\linewidth}
		\centering
		\includegraphics[scale = 0.5]{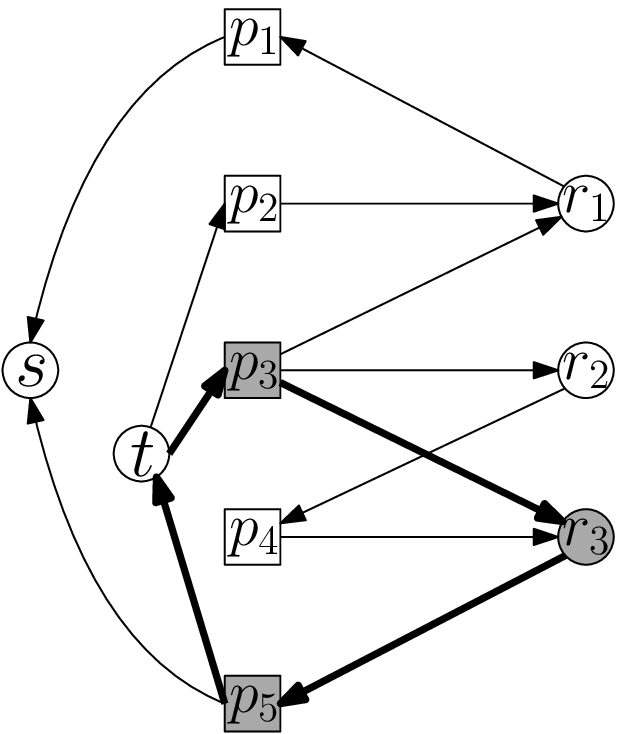}
		\caption{residual graph}
	\end{subfigure}
	\begin{subfigure}[t]{0.25\linewidth}
		\centering
		\includegraphics[scale = 0.5]{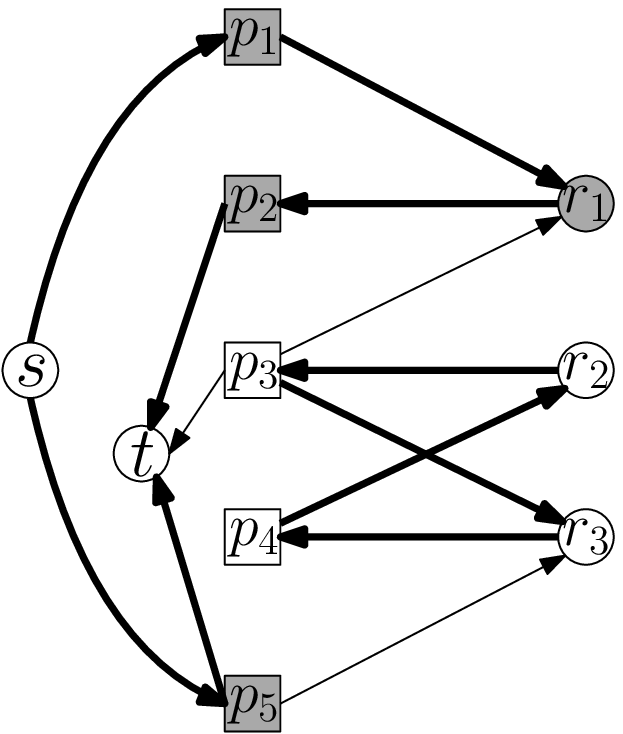}
		\caption{updated flow}
	\end{subfigure}
	\caption{An example of the proof of Claim~\ref{cl:path-reroute} with $S = \{p_1, p_5\}$
	and $T = \{p_2, p_3, p_5\}$}
	\label{fg:path-reroute}
\end{figure}

Figure~\ref{fg:path-reroute} gives an example of the proof of Claim~\ref{cl:path-reroute}.  In (a), $M$ consists of the bold edges. In (b), the bold edges form a flow $F$.  The bold edges other than those incident to $s$ and $t$ form a feasible solution $\Pi$ of $G_M(S,T)$ consisting of two paths: one from $p_1$ to $p_2$ and the other from $p_5$ to $p_3$.  Note that $\Pi^+ = \Pi$.  The residual graph of the flow network in (b) with respect to $F$ is shown in (c). The subgraph of the residual graph that excludes $s$ and $t$ is exactly $G_{M\oplus \Pi^+}$. The bold edges form a cycle $t\cdot \pi \cdot t$ where $\pi$ is a path in $G_{M\oplus \Pi^+}$. In (d), the bold edges form an $s$-$t$ flow $F'$ which is obtained from $F$ by pushing a unit flow along $t \cdot \pi \cdot t$. $F'$ induces a set $\Pi'$ of two node-disjoint paths from $S$ to $T$ in $G_M$: a trivial one from $p_5$ to itself and a non-trivial one from $p_1$ to $p_2$. The cycle $p_3r_3p_4r_2$ in $F'$ is ignored. $\Pi$ and $\Pi'$ satisfy Claim~\ref{cl:path-reroute}.

\subsubsection{More properties} 

We derive some relations between $f_M(S, T)$'s for different choices of $M$, $S$, and $T$.

\begin{claim}
	\label{cl:equal-path-num}
	For any maximum matchings $M$ and $M'$ of $G$, 
	\begin{emromani}
		\item $f_M(\overbar{P}_M,\overbar{P}_{M'}) = |\overbar{P}_{M}| = |\overbar{P}_{M'}|$, and
		\item for every subset $T$ of players, $f_M(\overbar{P}_M,T) = f_{M'}(\overbar{P}_{M'},T)$.
	\end{emromani}
\end{claim}
\begin{proof}
	We first prove (i).  Consider the symmetric difference $M\oplus M'$. It consists of cycles and alternating paths of even lengths~\cite{HK73}.  All these alternating paths are node-disjoint and appear as directed paths in $G_M$.  Since these paths have even lengths, they are either from players to players or from resources to resources.  Any node in $\overbar{P}_M \setminus \overbar{P}_{M'}$ (i.e., matched by $M'$ but not by $M$) must be an endpoint of some alternating path, and the other endpoint of the path must be a node in $\overbar{P}_{M'}\setminus \overbar{P}_{M}$ (i.e., matched by $M$ but not by $M'$ ). Any node in $\overbar{P}_M \cap \overbar{P}_{M'}$ has no incident edge in $M \oplus M'$, so it is a trivial path.  Putting things together, there are node-disjoint paths (trivial or non-trivial) in $G_M$ from all nodes in $\overbar{P}_M$ to $\overbar{P}_{M'}$.	 So $f_M(\overbar{P}_M, \overbar{P}_{M'}) = |\overbar{P}_M|$. 

	Next we prove (ii).  Let $\Pi$ be an optimal solution of $G_M(\overbar{P}_M,T)$.  Let $M''$ be the maximum matching obtained from $M$ by flipping the alternating paths in $\Pi^+$, i.e., $M'' = M \oplus \Pi^+$.  After flipping the alternating paths,  players in $\src(\Pi^+)$ become matched and players in $\sink(\Pi^+)$ become unmatched.  Thus, $\overbar{P}_{M''} = (\overbar{P}_M\setminus \src(\Pi^+)) \cup \sink(\Pi^+) = (\bar{P}_M\setminus \src(\Pi)) \cup \sink(\Pi)$. The last step is due to the fact that each trivial path is a single vertex in $\overbar{P}_M$ which serves as a source and a sink simultaneously.
	So $\sink(\Pi) \subseteq \overbar{P}_{M''}$.  By (i), $f_{M'}(\overbar{P}_{M'}, \overbar{P}_{M''}) = |\overbar{P}_{M'}| = |\overbar{P}_{M''}|$, which implies that $f_{M'}(\bar{P}_{M'}, \sink(\Pi)) = |\Pi|$  (Just take an optimal solution of $G_{M'}(\overbar{P}_{M'},\overbar{P}_{M''})$ and delete the paths ending at $\overbar{P}_{M''}\setminus \sink(\Pi)$).  As $\sink(\Pi) \subseteq T$ by definition, $f_{M'}(\overbar{P}_{M'},T) \geq f_{M'}(\overbar{P}_{M'},\sink(\Pi)) = |\Pi|$.  Recall that $\Pi$ is an optimal solution of $G_M(\overbar{P}_M,T)$, so $f_M(\overbar{P}_M, T) = |\Pi| \leq f_{M'}(\overbar{P}_{M'},T)$. We can similarly prove the other direction that $f_M(\overbar{P}_M,T) \geq f_{M'}(\overbar{P}_{M'},T)$.
\end{proof}

Claim~\ref{cl:more-path} below states that if adding a player $p$ to $T$ increases $f_M(S,T)$, adding $p$ to any subset $T' \subset T$ increases $f_M(S,T')$ too.

\begin{claim}
	\label{cl:more-path}
	Let $M$ be a maximum matching of $G$.  Let $S$ be any subset of $\overbar{P}_M$. Let $T$ be any subset of $P$. Let $p$ be an arbitrary player in $P$. If $f_M(S,T \cup \{p\}) = f_M(S,T) + 1$, then for every $T' \subseteq T$, $f_M(S,T'\cup \{p\}) = f_M(S,T') + 1$.
\end{claim}
\begin{proof}
	Let $\Pi_1$ be an optimal solution of $G_M(S,T')$.  Note that $\Pi_1$ is also a feasible solution of $G_M(S,T \cup \{p\})$.  Let $\Pi_2$ be an optimal solution of $G_M(S,T \cup \{p\})$ obtained by augmenting $\Pi_1$ (using Claim~\ref{cl:path-aug}).  Then, $\sink(\Pi_1) \subseteq \sink(\Pi_2)$.  If $p \in \sink(\Pi_2)$, then $\sink(\Pi_1) \cup \{p\} \subseteq \sink(\Pi_2)$, implying that there are $|\Pi_1| + 1 = f_M(S,T') + 1$ node-disjoint paths from $S$ to 
	$T' \cup \{p\}$, and thus establishing the claim.  If $p \not\in \sink(\Pi_2)$, then $\Pi_2$ is a feasible solution of $G_M(S,T)$.  But then $f_M(S,T \cup \{p\}) = |\Pi_2| \leq f_M(S,T)$, a contradiction to the assumption.
\end{proof}

%% file: src/algorithm.tex
\section{The Algorithm}
\label{sec:alg}

In this section, we present an algorithm which, given a partial allocation and an unsatisfied player $p_0$, computes a new partial allocation that satisfies $p_0$ and all the players that used to be satisfied. Recall that a partial allocation consists of a maximum matching $M$ of $G$ and a subset $\cal E$ of thin edges such that (i) no two edges in $M$ and $\cal E$ satisfy (i.e., cover) the same player, (ii) no two edges in $\cal E$ share any resource, (iii) every edge $(p,B) \in {\cal E}$ is minimal in the sense that every proper subset $B' \subset B$ has value less than $\tau/\lambda$.

Let $M$ and ${\cal E}$ be the maximum matching of $G$ and the set of thin edges in the current partial allocation, respectively.  Let $p_0$ be an arbitrary player who is not yet satisfied.

\subsection{Overview}
\label{sec:overview}
To satisfy $p_0$, the simplest case is that we can find a minimal thin edge $(p_0, B_0)$ such that $B_0$ excludes all the resources covered by ${\cal E}$. (Recall that by definition of thin edges, $\val(B_0) \geq \tau/\lambda$.)  We can extend the partial allocation by adding $(p_0, B_0)$ to ${\cal E}$.

More generally, we can use any thin edge $(q_0, B_0)$ such that $B_0$ meets the above requirements even if $q_0 \not= p_0$, provided that there is a path from $p_0$ to $q_0$ in $G_M$.  If $q_0 \neq p_0$, such a path is an alternating path in $G$ with respect to $M$, and $q_0$ is matched by $M$.  We can flip this path to match $p_0$ with a fat resource and then include $(q_0,B_0)$ in ${\cal E}$ to satisfy $q_0$.

The thin edge $(q_0,B_0)$ mentioned above may not always exist. In other words, some edges in ${\cal E}$ may share resources with $(q_0, B_0)$.  Let $\{(p_1,B'_1),\ldots,(p_k,B'_k)\}$ be such thin edges in $\cal E$.  In this situation, we say $(q_0, B_0)$ is blocked by $\{(p_1,B'_1),\ldots,(p_k,B'_k)\}$.  In order to free up the resources held by $p_1,\ldots,p_k$ and to make $(q_0,B_0)$ unblocked, we need to satisfy each player $p_i$ with resources other than those in $B_0, B'_1,\ldots, B'_k$.  Afterwards, we can satisfy $p_0$ as before.  To record the different states of the algorithm, we initialize a stack to contain $(p_0,\emptyset)$ as the first layer and then create another layer on top that stores the sets ${\cal X}_2 = \{(q_0,B_0)\}$ and ${\cal Y}_2 = \{(p_1,B'_1),\ldots,(p_k,B'_k)\}$ among other things for bookkeeping.  We change our focus to satisfy the set of players $Y_2 = \{p_1,\ldots,p_k\}$.

To satisfy a player in $Y_2$ (by a new edge), we need to identify a minimal thin edge $(q_1,B_1)$ such that $B_1$ excludes the resources already covered by thin edges in the current stack because we don't want $(q_1, B_1)$ to block or be blocked by any edges in the current stack.  As to $q_1$, we require that $G_M$ contains two node-disjoint paths from $\{p_0\} \cup Y_2$ to $\{q_0,q_1\}$.  If $(q_1,B_1)$ is blocked by some thin edges in ${\cal E}$, we initialize a set ${\cal X}_3 = \{(q_1,B_1)\}$; otherwise, we initialize a set ${\cal I} = \{(q_1,B_1)\}$.  Ideally, if $(q_1,B_1)$ is unblocked, we could immediately make some progress.  Since there are two node-disjoint paths from $\{p_0\} \cup Y_2$ to $\{q_0,q_1\}$, $q_1$ is reachable from either $p_0$ or a player in $Y_2$.  In the former case, we can satisfy $p_0$. In the latter case, the path from $Y_2$ to $q_1$ must be node-disjoint from the path from $p_0$ to $q_0$, so we can satisfy some player $p_i$ without affecting the alternating path from $p_0$ to $q_0$, and free up the resources previously held by that player.  But we would not do so because, as argued in~\cite{AKS17}, in order to achieve a polynomial running time, we should let $\cal I$ grow bigger so that a larger progress can be made at once.

Since there are multiple players in $Y_2$ to be satisfied, we continue to look for another minimal thin edge $(q_i,B_i)$.  We require that $B_i$ excludes the resources covered by thin edges in the current stack (including ${\cal X}_3$ and $\cal I$), and that $G_M$ contains one more  node-disjoint paths from $\{p_0\} \cup Y_2$ after adding $q_i$ to the destination set.  If $(q_i,B_i)$ is blocked by some thin edges in ${\cal E}$, we add $(q_i,B_i)$ to ${\cal X}_3$; otherwise, we add it to ${\cal I}$.  After collecting all such thin edges in ${\cal X}_3$ and $\cal I$, we construct the set ${\cal Y}_3$ of thin edges in the current partial allocation that block ${\cal X}_3$.  Then, we add a new top layer to the stack that stores ${\cal X}_3$ and ${\cal Y}_3$  among other things for bookkeeping.  Then, in order to free up the resources held by edges in ${\cal Y}_3$ and to make edges in ${\cal X}_3$ unblocked, we turn our attention to satisfying the players covered by ${\cal Y}_3$ with new edges and so on.  These repeated additions of layers to the stack constitute the build phase of the algorithm.

The build phase stops when we have enough thin edges in ${\cal I}$ to satisfy a predetermined fraction of players covered by ${\cal Y}_l$ for some $l$. We shrink the $l$-th layer and delete all layers above it.  The above is repeated until ${\cal I}$ is not large enough to satisfy the predetermined fraction of players covered by any ${\cal Y}_l$ in the stack.  These repeated removal of layers constitute the collapse phase of the algorithm.  At the end of the collapse phase, we switch back to the build phase.

The alternation of build and collapse phases continues until we succeed in satisfying player $p_0$, our original goal, which is stored in the bottommost layer in the stack.

The lazy update (i.e., wait until $\cal I$ is large enough before switching to the collapse phase) is not sufficient for achieving a polynomial running time. A greedy strategy is also needed.  In~\cite{AKS17}, when a blocked thin edge $(q,B)$ is picked and added to ${\cal X}_l$ for some $l$, $B$ is required to be a minimal set of value at least $\tau/2$, which is more than $\tau/\lambda$.  Intuitively, if such an edge is blocked, it must be blocked by many edges. Hence, the strategy leads to a fast growth of the stack.  We use a more aggressive strategy: we allow the value of $B$ to be as large as $\tau+\tau/\lambda$, and among all candidates, we pick the thin edge $(q,B)$ with (nearly) the largest value.  Our strategy leads to a faster growth of the stack, and hence, a polynomial running time can be achieved for a smaller $\lambda$.

\subsection{Notation and definitions}  

Let $M$ and $\cal E$ denote the maximum matching of $G$ and the set of thin edges, respectively, that are used in the current partial allocation.  Let $p_0$ denote the next player we want to satisfy.  

A \emph{state of the algorithm} consists of several components, namely, $M$, $\cal E$, a stack of layers, and a global variable $\cal I$ that stores a set of unblocked thin edge.  The layers in the stack are indexed starting from 1 at the bottom.  For $i \geq 1$, the $i$-th layer $L_i$ is a 4-tuple $({\cal X}_i, {\cal Y}_i,d_i,z_i)$, where ${\cal X}_i$ and ${\cal Y}_i$ are sets of thin edges, and $d_i$ and $z_i$ are two numeric values that we will explain later.  We use $I$, $X_i$ and $Y_i$ to denote the sets of players covered by edges in $\cal I$, ${\cal X}_i$ and ${\cal Y}_i$, respectively.  The set $\cal I$ grows during the build phase and shrinks during the collapse phase, and $I$ changes correspondingly.  The same is true for ${\cal X}_i$, $X_i$, ${\cal Y}_i$, and $Y_i$.  For any $k \geq 1$, let ${\cal X}_{\leq k}$ denote $\bigcup_{i=1}^{k} {\cal X}_i$. ${\cal Y}_{\leq k}$, $X_{\leq k}$, and $Y_{\leq k}$ are similarly defined.

The sets ${\cal X}_i$ and ${\cal Y}_i$ are defined inductively.  At the beginning of the algorithm, ${\cal X}_1 = \emptyset$, ${\cal Y}_1 = \{(p_0,\emptyset)\}$, $d_1 = z_1 = 0$, and ${\cal I} = \emptyset$.  The first layer in the stack is thus $(\emptyset,\{(p_0,\emptyset)\},0,0)$.

Let $\ell$ be the index of the topmost layer in the stack. Consider the construction of the $(\ell+1)$-th layer in an execution of the build phase.  When it first starts, ${\cal X}_{\ell+1}$ is initialized to be empty.  We say that a player $p$ is {\bf addable} if
\[
	f_M(Y_{\leq \ell}, X_{\leq \ell+1} \cup I \cup \{p\}) = f_M(Y_{\leq \ell}, X_{\leq \ell+1} \cup I) + 1.
\]
Note that this definition depends on $X_{\leq \ell+1} \cup I$, so adding edges to ${\cal X}_{\ell+1}$ and $\cal I$ may affect the addability of players.

Given an addable player $p$, we say that a thin edge $(p,B)$ is {\bf addable} if
\[
	\val(B) \in [\tau/\lambda,\tau + \tau/\lambda] \,\, \mbox{and} \,\, \mbox{$B$ excludes resources covered by ${\cal X}_{\leq \ell+1} \cup {\cal Y}_{\leq \ell} \cup {\cal I}$}.
\]
An addable thin edge $(p,B)$ is {\bf unblocked} if there exists a subset $B' \subseteq B$ such that $\val(B') \geq \tau/\lambda$ and $B'$ excludes resources used in $\cal E$.  Otherwise, $(p,B)$ is {\bf blocked}.  During the construction of the $(\ell+1)$-th layer, the algorithm adds some blocked addable thin edges to ${\cal X}_{\ell+1}$ and some unblocked addable thin edges to ${\cal I}$.  When the growth of ${\cal X}_{\ell+1}$ stops, the algorithm constructs ${\cal Y}_{\ell+1}$ as the set of thin edges in $\cal E$ that share resource(s) with some edge(s) in ${\cal X}_{\ell+1}$.  

After constructing ${\cal X}_{\ell+1}$ and ${\cal Y}_{\ell+1}$ and growing $\cal I$, the values $d_{\ell+1}$ and $z_{\ell+1}$ are defined as
\[
	d_{\ell+1} := f_M(Y_{\leq \ell}, X_{\leq \ell+1} \cup I), \quad
	z_{\ell+1} := |X_{\ell+1}|.
\]
The values $d_{\ell+1}$ and $z_{\ell+1}$ do not change once computed unless the layer $L_{\ell+1}$ is destructed in the collapse phase.  That is, $d_{\ell+1}$ and $z_{\ell+1}$ record the values $f_M(Y_{\leq \ell}, X_{\leq \ell+1} \cup I)$ and $|X_{\ell+1}|$ at the time of construction.  (Note that $f_M(Y_{\leq \ell}, X_{\leq \ell+1} \cup I)$ and $|X_{\ell+1}|$ may change subsequently.)  The values $d_{\ell+1}$ and $z_{\ell+1}$ are introduced only for the analysis.  They are not used by the algorithm.

Whenever we complete the construction of a new layer in the stack, we enter the collapse phase to  check whether any existing layer is collapsible.  If so, shrink the stack and update the current partial allocation ($M$ and $\cal E$).  We stay in the collapse phase until no layer is collapsible.  If the stack has become empty, we are done as the player $p_0$ has been satisfied.  Otherwise, we reenter the build phase.  We give the detailed specification of the build and collapse phases in the following subsections.

\subsection{Build phase}  

Let $\ell$ be the index of the topmost layer in the stack.  Let $M$ and $\cal E$ denote the maximum matching in $G$ and the set of thin edges in the current partial allocation, respectively.  We call the following routine {\sc Build} to construct the next layer $L_{k+1}$.

\begin{quote}
\noindent {\sc Build}$(M, {\cal E}, {\cal I}, (L_1,\cdots,L_\ell))$
\begin{enumerate}
	\item Initialize ${\cal X}_{\ell+1}$ to be the empty set.  \label{step:build-1}
	
	\item If there is an addable player $p$ and an unblocked addable edge $(p,B)$, then: \label{step:build-2}
	\begin{enumerate}
		\item take a minimal subset $B' \subseteq B$ such that $\val(B')\geq \tau/\lambda$ and $B'$ excludes the resources covered by $\cal E$ (we call $(p,B')$ a \emph{minimal unblocked addable edge}),
		\item add $(p,B')$ to $\cal I$,
		\item repeat step~\ref{step:build-2}. 
	\end{enumerate}

	\item When we come to step~\ref{step:build-3}, no unblocked addable edge is left.  If there is no (blocked) addable edge, go to step~\ref{step:build-4}.  For each addable player $p$ who is incident to at least one addable edge, identify one \emph{maximal blocked addable edge $(p,B)$} such that $B \not\subset B'$ for any blocked addable edge $(p,B')$. Among the maximal blocked addable edges identified, pick the one with the largest value, and add it to ${\cal X}_{\ell+1}$. Then repeat step~\ref{step:build-3}. \label{step:build-3}
	
	\item At this point, the construction of ${\cal X}_{\ell+1}$ is complete.  Let ${\cal Y}_{\ell+1}$ be the set of the thin edges in $\cal E$ that share resource(s) with some thin edge(s) in ${\cal X}_{\ell+1}$. \label{step:build-4}
	
	\item Compute $d_{\ell+1} := f_M(Y_{\leq \ell}, X_{\leq \ell+1} \cup I)$ and $z_{\ell+1} := \bigl|X_{\ell+1}\bigr|$. \label{step:build-5}
	
	\item Push the new layer $L_{\ell+1} = ({\cal X}_{\ell+1},{\cal Y}_{\ell+1},d_{\ell+1},z_{\ell+1})$ onto the stack. $\ell := \ell + 1$. \label{step:build-6}
\end{enumerate}
\end{quote}

{\sc Build} differs from its counterpart in~\cite{AKS17} in several places, particularly in step~\ref{step:build-3}.  First, we require blocked addable edges to be maximal while only minimal addable edges of value at least $\tau/2$ are considered in~\cite{AKS17}.  Second, when adding addable edges to ${\cal X}_{\ell+1}$, we pick the one with (nearly) the largest value.  In contrast, one arbitrary addable edge is picked in~\cite{AKS17}. 

One may wonder, instead of identifying a maximal blocked addable edge for each player, whether it is better to identify a maximum blocked addable edge (i.e., the blocked addable edge with the largest value).  However, finding the blocked addable edge with the largest value for $p$ is an instance of the NP-hard knapsack problem.  Maximal blocked addable edges are sufficient for our purposes.

Is it possible that ${\cal X}_{\ell+1} = \emptyset$ and so ${\cal Y}_{\ell+1} = \emptyset$?  We will establish the result, Lemma~\ref{lem:key} in Section~\ref{sec:key}, that if $|Y_{i+1}| < \sqrt{\mu}|Y_{\leq i}|$ for some $i$, then some layer below $L_{i+1}$ is collapsible.  Therefore, if ${\cal Y}_{\ell+1}$ is empty, then some layer below $L_{\ell+1}$ must be collapsible, the algorithm will enter the collapse phase next, and $L_{\ell+1}$ will be removed.

\begin{lemma}\label{lem:time-build}
	{\sc Build} runs in $\poly(m,n)$ time.
\end{lemma}
\begin{proof}
	It suffices to show that steps~\ref{step:build-2} and \ref{step:build-3} run in polynomial time.  Two maximum flow computations tell us whether a player $p$ is addable.  Suppose so.  We start with the thin edge $(p,B)$ where $B = \emptyset$.  Let $R_p$ denote the set of thin resources that are desired by $p$.  First, we incrementally insert to $B$ thin resources from $R_p$ that appear in neither the current partial allocation nor $\mathcal{X}_{\leq \ell+1} \cup \mathcal{Y}_{\leq \ell} \cup \mathcal{I}$.  If $\val(B)$ becomes greater than or equal to $\tau/\lambda$, then we must be in step~\ref{step:build-2} and $(p,B)$ is a minimal unblocked addable edge that can be added to $\cal I$.  Suppose that $\val(B) < \tau/\lambda$ after the incremental insertion stops.  Then, $p$ has no unblocked addable edge.  If we are in step~\ref{step:build-3}, we continue to add to $B$ thin resources from $R_p$ that appear in the current partial allocation but not in $\mathcal{X}_{\leq \ell+1}\cup \mathcal{Y}_{\leq \ell} \cup \mathcal{I}$.  If $\val(B) < \tau/\lambda$ when the incremental insertion stops, then $p$ has no addable edge.  Otherwise, we continue until $\val(B)$ is about to exceed $\tau + \tau/\lambda$ or we have examined all thin resources in $R_p$, whichever happens earlier.  In either case, the final $\val(B)$ is in the range $\left[\tau/\lambda, \tau+\tau/\lambda\right]$ and $(p,B)$ is a maximal blocked addable edge.
\end{proof}

Table~\ref{tb:invar} shows the invariants that will be used in the analysis of the algorithm. Clearly, they all hold at the start of the algorithm, i.e., $\ell = 1$, ${\cal X}_1 = \emptyset$, ${\cal I} = \emptyset$, ${\cal Y}_1 = \{(p_0,\emptyset)\}$, and $d_1 = z_1 = 0$. We show that they are maintained by {\sc Build}.

\begin{table}
	\centering
	\caption{
		Invariants maintained by the algorithm. $(M, {\cal E}, {\cal I}, \{L_1, \ldots, L_{\ell}\})$ is the current state of the algorithm.}
	\label{tb:invar}
	\vspace{.1in}
	\begin{tabu} to \textwidth {X[-.5cb]X[lp]}
	\toprule
	\label{invar:1}
	Invariant~1  & Every edge in $\cal I$ has value in the range $[\tau/\lambda,2\tau/\lambda]$. Every edge in ${\cal X}_{\leq \ell}$ has value in the range $[\tau/\lambda,\tau+\tau/\lambda]$.  No two edges from ${\cal X}_{\leq \ell} \cup {\cal I}$ cover the same player or share any resource. \\
	\midrule
	\label{invar:2}
	Invariant~2 & No edge in $\cal E$ shares any resource with any edge in $\cal I$.  \\
	\midrule
	\label{invar:3}
	Invariant~3 & For all $i \in [1,\ell]$, every edge in ${\cal X}_i$ shares some resource(s) with some edge(s) in ${\cal Y}_i$ but not with any edge in ${\cal E} \setminus {\cal Y}_i$. \\
	\midrule
	\label{invar:4}
	Invariant~4 &  ${\cal Y}_2, \ldots, {\cal Y}_\ell$ are disjoint subsets of $\cal E$. (Note that ${\cal Y}_1 = \{(p_0, \emptyset)\}$ is not a subset of $\cal E$ and is certainly disjoint from ${\cal Y}_2, \ldots, {\cal Y}_\ell$.)\\
	\midrule
	\label{invar:5}
	Invariant~5 & For all $i \in [1,\ell]$, no edge in ${\cal Y}_i$ shares any resource with any edge in ${\cal X}_j$ for any $j \not= i$. \\
	\midrule
	\label{invar:6}
	Invariant~6 & $f_M(Y_{\leq \ell-1},I) = |I|$. \\
	\midrule
	\label{invar:7}
	Invariant~7 & For all $i \in [1,\ell-1]$, $f_M(Y_{\leq i},X_{\leq i+1} \cup I) \geq d_{i+1}$. \\
	\bottomrule
	\end{tabu}
\end{table} 

\begin{lemma}
	{\sc Build} maintains invariants~1--7 in Table~\ref{tb:invar}.
\end{lemma}
\begin{proof}
	Suppose that the invariants hold before {\sc Build} constructs the new topmost layer $L_{\ell+1}$.  It suffices to check the invariants after the construction of $L_{\ell+1}$.  Invariants~1 and 2 are clearly preserved by the working of {\sc Build}. 

	Consider invariant~3.  It holds for $i \in [1,\ell]$ because none of $\cal E$, ${\cal X}_i$, and ${\cal Y}_i$ is changed.  Since all edges in $\cal E$ that share some resource(s) with some edge(s) in ${\cal X}_{\ell+1}$ are added to ${\cal Y}_{\ell+1}$, invariant~3 also holds for $i = \ell+1$.

	Consider invariant~4.   By induction assumption, ${\cal Y}_2,\ldots,{\cal Y}_\ell$ are disjoint subsets of $\cal E$.  By construction, ${\cal Y}_{\ell+1} \subseteq {\cal E}$.  If an edge $e \in {\cal Y}_{\ell+1}$ belongs to ${\cal Y}_{\leq \ell}$, then the resources covered by $e$ must be excluded by all edges in ${\cal X}_{\ell+1}$ by the definition of addable edges.  But $e$ must share resource(s) with some edge(s) in ${\cal X}_{\ell+1}$ in order that $e \in {\cal Y}_{\ell+1}$, a contradiction.  So ${\cal Y}_2,\ldots, {\cal Y}_{\ell+1}$ are disjoint subsets of $\cal E$.

	Invariant~5 follows from invariants~3 and~4 and the fact that edges in $\cal E$ do not share any resource.

	Let ${\cal I}'$ denote the version of $\cal I$ immediately before the execution of {\sc Build} to construct $L_{\ell+1}$.  Let ${\cal I}''$ denote the set of unblocked addable edges inserted into ${\cal I}'$ for constructing $L_{\ell+1}$.  Correspondingly, $I'$ and $I''$ denote the set of players covered by ${\cal I}'$ and ${\cal I}''$, respectively.

	Consider invariant~6.  We have $f_M(Y_{\leq \ell-1},I') = |I'|$ as invariant~6 is assumed to hold before {\sc Build} executes.   Each player in $X_{\ell+1}$ and $I''$ is determined to be addable, i.e., adding such a player to $X_{\leq \ell+1} \cup I$ increases the value of $f_M(Y_{\leq \ell}, X_{\leq \ell+1} \cup I)$ by one.  Then, Claim~\ref{cl:more-path} implies that $f_M(Y_{\leq \ell},I'\cup I'') = f_M(Y_{\leq \ell},I') + |I''|$.  Recall that $f_M(Y_{\leq \ell-1},I') = |I'|$, which implies $f_M(Y_{\leq \ell},I') = |I'|$.  Thus $f_M(Y_{\leq \ell}, I' \cup I'') = |I'| + |I''| = |I' \cup I''|$, preserving invariant~6.

	Consider invariant~7.  {\sc Build} does not change $Y_{\leq i}$ and $X_{\leq i+1}$ for any $i \in [1,\ell-1]$, and {\sc Build} does not delete any edge from $\cal I$.  Therefore, for all $i \in [1,\ell-1]$, $f_M(Y_{\leq i},X_{\leq i+1} \cup I)$ cannot decrease and so $f_M(Y_{\leq i},X_{\leq i+1} \cup I)$ remains larger than or equal to $d_{i+1}$.  By construction, {\sc Build} sets $d_{\ell+1} := f_M(Y_{\leq \ell}, X_{\leq \ell+1} \cup I)$.  Invariant 7 is preserved.
\end{proof}

\subsection{Collapse phase}

Let $M$ be the maximum matching in $G$ in the current partial allocation.  Let  $L_1, L_2, \ldots, L_\ell$ be the layers currently in the stack from bottom to top.  We need to tell whether a layer can be collapsed, and this requires a certain decomposition of $\cal I$.

\paragraph*{Collapsibility.}  Let ${\cal I}_1 \cup {\cal I}_2 \cup \cdots {\cal I}_\ell$ be some partition of $\cal I$.  Let $I_i$ denote the set of players covered by ${\cal I}_i$.  We use ${\cal I}_{\leq j}$ and $I_{\leq j}$ to denote $\bigcup_{i=1}^j {\cal I}_i$ and $\bigcup_{i=1}^j I_i$, respectively.  Note that $|I_i| = |{\cal I}_i|$ by invariant~1 in Table~\ref{tb:invar}.  The partition ${\cal I}_1 \cup {\cal I}_2 \cup \cdots {\cal I}_\ell$ is a \emph{canonical decomposition} of $\cal I$~\cite{AKS17} if
\[
\forall \, i \in [1,\ell], \quad f_M(Y_{\leq i}, I_{\leq i}) = f_M(Y_{\leq i}, I) = |I_{\leq i}| = |{\cal I}_{\leq i}|.
\]

\begin{lemma}
\label{lem:canon}
	In $\poly(\ell,m,n)$ time, one can compute a canonical decomposition 
		${\cal I}_1 \cup {\cal I}_2 \cup \ldots {\cal I}_\ell$ of $\cal I$ 
	and a canonical solution of $G_M(Y_{\leq \ell},I)$ that can be partitioned into a disjoint union $\Gamma_1 \cup \Gamma_2 \cup \cdots \Gamma_\ell$ such that for every $i \in [1,\ell]$, $\Gamma_i$ is a set of $|I_i|$ paths from $Y_i$ to $I_i$.  
\end{lemma}
\begin{proof}
	We first compute an optimal solution $\Pi_1$ of $G_M(Y_1,I)$ by successive augmentations (using Claim~\ref{cl:path-aug}).  So $|\Pi_1| = f_M(Y_1,I)$.  For $j = 2,\ldots,\ell$, we compute an optimal solution $\Pi_j$ of $G_M(Y_{\leq j},I)$ by successively augmenting $\Pi_{j-1}$.  By Claim~\ref{cl:path-aug}, $\src(\Pi_{j-1}) \subseteq \src(\Pi_j)$.  Therefore, we inductively maintain the property that for all $i \in [1,j]$, $\Pi_j$ contains $|\Pi_i| = f_M(Y_{\leq i},I)$ node-disjoint paths from $Y_{\leq i}$ to $I$.  In the end, $I = \sink(\Pi_\ell)$ by invariant~6 in Table~\ref{tb:invar}.  We obtain the canonical decomposition and canonical solution as follows: for every $i \in [1,\ell]$, let $\Gamma_i$ be the subset of paths in $\Pi_\ell$ from $Y_i$ to $I$, let $I_i = \sink(\Gamma_i)$, and let ${\cal I}_i$ be the subset of edges in $\cal I$ that cover the players in $I_i$.  The $\Gamma_i$'s are disjoint because the $Y_i$'s are disjoint by invariant~4 in Table~\ref{tb:invar}.
\end{proof}

Note that $I_\ell$ and $\Gamma_\ell$ are actually empty because the invariant 6 in Table~\ref{tb:invar} ensures that $f_M(Y_{\leq \ell-1},I) = |I|$ when we enter the collapse phase.

Consider $\Gamma_i$. The sources (which are also sinks) of the trivial paths in $\Gamma_i$ are players covered by unblocked addable thin edges in ${\cal I}_i$, and can be satisfied by these thin edges. Recall that the non-trivial paths in $\Gamma_i$ are alternating paths with respect to $M$. If we flip these alternating paths, their sources can be satisfied by fat resources. Their  sinks can be satisfied by thin edges in ${\cal I}_i$. The subset of ${\cal Y}_i$ that cover $\src(\Gamma_i)$ should then be removed from ${\cal E}$ because the players in $\src(\Gamma_i)$ are now satisfied by either a fat edge or a thin edge from ${\cal I}_i$. This subset of ${\cal Y}_i$ should also be removed from ${\cal Y}_i$ because they no longer block edges in ${\cal X}_i$. A layer is collapsible if a certain portion of its blocking edges can be removed.  More precisely, for any $i \in [1,\ell]$, $L_i$ is {\bf collapsible}~\cite{AKS17} if:
\begin{quote}
	$\exists$ a canonical decomposition 
		${\cal I}_1 \cup {\cal I}_2 \cup \ldots {\cal I}_\ell$ of $\cal I$ 
	such that $|I_i| \geq \mu |Y_i|$, where $\mu$ is a constant that will be defined later.
\end{quote}
By Lemma~\ref{lem:canon}, the collapsibility of a layer can be checked in $\poly(\ell,m,n)$ time.

\paragraph{Collapse layers.}  Let $(L_1,\ldots,L_\ell)$ denote the current layers in the stack.  The routine {\sc Collapse} below checks whether some layer is collapsible, and if yes, it collapses layers in the stack until no collapsible layer is left.  The execution of {\sc Collapse} may update the stack, ${\cal I}$, and the current partial allocation, including both the maximum matching $M$ in $G$ and the set of thin edges $\cal E$ in the partial allocation.  {\sc Collapse} works in the same manner as its counterpart in~\cite{AKS17}, but there are small differences in the presentation.

\begin{quote}
\noindent {\sc Collapse}$(M, {\cal E}, {\cal I}, (L_1,\cdots,L_\ell))$
\begin{enumerate}
	\item Compute a canonical decomposition ${\cal I}_1 \cup {\cal I}_2 \cup \cdots {\cal I}_\ell$ of ${\cal I}$ and a canonical solution $\Gamma_1 \cup \Gamma_2 \cup \cdots \Gamma_\ell$ of $G_M(Y_{\leq \ell},I)$.  If no layer is collapsible, go to build phase.  Otherwise, let $L_t$ be the collapsible layer with the smallest index $t$.\label{step:collapse-1}

	\item Remove all layers above $L_t$ from the stack.  Set ${\cal I} := {\cal I}_{\leq t-1}$. \label{step:collapse-2}

	\item Recall that $\src(\Gamma_t) \subseteq Y_t$ by Lemma~\ref{lem:canon}.  Let ${\cal V}$ denote the subset of edges in ${\cal Y}_t$ that cover players in $\src(\Gamma_t)$. \label{step:collapse-3}
	\begin{enumerate}
		\item Update the maximum matching $M$ by flipping the non-trivial paths in $\Gamma_t$, i.e., set $M := M \oplus \Gamma^+_{t}$ where $\Gamma^+_{t}$ is the set of non-trivial paths in $\Gamma_{t}$.  This update matches the sources of non-trivial paths in $\Gamma_t$, and makes their sinks unmatched. 
	
		\item Add to ${\cal E}$ the edges in ${\cal I}_t$, i.e., set ${\cal E} := {\cal E} \cup {\cal I}_t$. The sinks of non-trivial paths in $\Gamma_t$ and the sources (which are also sinks) of trivial paths $\Gamma_t$ are now satisfied by edges in ${\cal I}_t$.

		\item Each player in $\src(\Gamma_t)$ is now satisfied by either a fat resource or a thin edge from ${\cal I}_t$. If $t = 1$, then $p_0$ is already satisfied, and the algorithm terminates.  Assume that $t \geq 2$.  Then ${\cal V} \subseteq {\cal Y}_t \subseteq {\cal E}$. Edges in $\cal V$ can be removed from $\cal E$, so set ${\cal E} := {\cal E} \setminus {\cal V}$.  Consequently, edges in $\cal V$ no longer block edges in ${\cal X}_t$, so set ${\cal Y}_t := {\cal Y}_t \setminus {\cal V}$.  
	\end{enumerate}

	\item If $t \geq 2$, we need to update ${\cal X}_t$ because the removal of $\cal V$ from ${\cal E}$ (and hence ${\cal Y}_t$) may make some edges in ${\cal X}_t$ unblocked.  For each edge $(p,B) \in {\cal X}_t$ that becomes unblocked, perform the following operations: \label{step:collapse-4}
	\begin{enumerate}
		\item Remove $(p,B)$ from ${\cal X}_t$.
		
		\item If $f_M(Y_{\leq t-1}, I\cup \{p\}) = f_M(Y_{\leq t-1}, I) + 1$, then add $(p,B')$ to $\cal I$, where $B'$ is an arbitrary minimal subset of $B$ such that $\val(B') \geq \tau/\lambda$ and $B'$ excludes the resources used in $\cal E$.
	\end{enumerate}  

	\item If $t = 1$, step~\ref{step:collapse-3} already satisfied the player $p_0$ in the bottommost layer in the stack, so the algorithm terminates.  Otherwise, update $\ell := t$ and go back to step~\ref{step:collapse-1}.
\end{enumerate}
\end{quote}

When we shrink a layer in the stack, the basis for generating the layers above it is no longer valid, and therefore, the easiest handling is to remove all such layers.  In turn, it means that we should shrink the lowest collapsible layer $L_t$ in the stack, which explains the choice of $t$ in step~\ref{step:collapse-1}.  In one iteration, the layers $L_1,\ldots,L_{t-1}$ are preserved, $L_t$ is updated, and the updated $L_t$ becomes the topmost layer in the stack.  One may wonder if it is possible that ${\cal Y}_t$ becomes empty after step~\ref{step:collapse-3} and so ${\cal X}_t$ also becomes empty after step~\ref{step:collapse-4}.  We will establish the result, Lemma~\ref{lem:key} in Section~\ref{sec:key}, that if $|Y_{i+1}| < \sqrt{\mu}|Y_{\leq i}|$ for some $i$, then some layer below the $(i+1)$-th layer is collapsible.  Therefore, if ${\cal Y}_t$ becomes empty, some layer below $L_t$ is collapsible and so $L_t$ will be removed in the next iteration of steps~\ref{step:collapse-2}--\ref{step:collapse-4} of {\sc Collapse}.

We need to show that invariants~1--7 in table~\ref{tb:invar} are satisfied after {\sc Collapse} terminates so that we are ready to enter the build phase.

\begin{lemma}
	{\sc Collapse} maintains invariants~1--7 in Table~\ref{tb:invar}.
\end{lemma}
\begin{proof}
	It suffices to show that invariants~1--7 are preserved after collapsing the lowest collapsible layer $L_t$ in steps~\ref{step:collapse-2}--\ref{step:collapse-4}.  Clearly, invariant~1 is preserved by the working of {\sc Collapse}.  Consider invariant~2. Since $\cal{I}$ is updated to ${\cal I}_{\leq t-1}$, the inclusion of ${\cal I}_t$ into ${\cal E}$ in step~\ref{step:collapse-3}(b) does not break invariant~2.  Step~\ref{step:collapse-3}(c) deletes edges from $\cal E$, which clearly preserves invariant~2.  In step~\ref{step:collapse-4}(b), all edges added to $\cal I$ exclude the resources used by $\cal E$, so invariant~2 is preserved.  Let $M$ denote the maximum  matching in the current partial allocation.  Let $(L_1,L_2,\ldots,L_\ell)$ be the layers in the stack before step~2.  It suffices to consider $(L_1,L_2,\ldots,L_t)$ for invariants~3--7 as $L_t$ will become stack top after one iteration of steps~\ref{step:collapse-2}--\ref{step:collapse-4}.

	Consider invariant~3.  Only steps~\ref{step:collapse-3}(b) and~\ref{step:collapse-3}(c) may affect it.  In step~\ref{step:collapse-3}(b), the edges that are added to $\cal E$ are from $\cal I$.   By invariant~1, no edge from $\cal I$ shares any resource with any edge in ${\cal X}_{\leq t}$, so invariant~3 is preserved.  In step~\ref{step:collapse-3}(c), we may delete edges from ${\cal Y}_t$, but these edges are also deleted from $\cal E$, so invariant~3 still holds.  
	
	Invariant~4 holds because {\sc Collapse} never adds any new edge to any ${\cal Y}_i$, and when some edges are removed from $\cal E$ by {\sc Collapse}, they are also removed from ${\cal Y}_t$.
	
	Invariant~5 follows from invariants~3 and~4 and the fact that edges in $\cal E$ do not share any resource.
	
	Consider invariant~6.  Since the topmost layer in the stack is going to be $L_t$, invariant~6 is concerned with $f_M(Y_{\leq t-1},I)$.  By the definition of a canonical solution, $\Gamma_{t-1}$ certifies that 
		$f_M(Y_{\leq t-1},I_{\leq t-1}) = |I_{\leq t-1}| \iff f_M(Y_{\leq t-1}, I) = |I|$ 
	as ${\cal I} := {\cal I}_{\leq t-1}$ in step~\ref{step:collapse-2}.  In step~\ref{step:collapse-3}, only step~\ref{step:collapse-3}(a) may affect the equality $f_M(Y_{\leq t-1},I) = |I|$ because $M$ may be changed by flipping the alternating paths in $\Gamma^+_t$.   Let $M'$ denote the updated matching.  $\Gamma_t^+$ is node-disjoint from $\Gamma_{\leq t-1}$, so flipping the alternating paths in $\Gamma^+_t$ does not affect $\Gamma_{\leq t-1}$.   Therefore, $\Gamma_{\leq t-1}$ still certifies that 
		$f_{M'}(Y_{\leq t-1},I_{\leq t-1}) = |I_{\leq t-1}| \iff f_{M'}(Y_{\leq t-1},I) = |I|$, 
	so step~\ref{step:collapse-3} preserves invariant~6.  In step~\ref{step:collapse-4}, a new edge $(p,B)$ is inserted into $\cal I$ only if 
		$f_{M'}(Y_{\leq t-1},I \cup \{p\}) > f_{M'}(Y_{\leq t-1},I)$.  
	Thus, when the size of $\cal I$ increases by one, $f_{M'}(Y_{\leq t-1},I)$ also increases by one.  So step~\ref{step:collapse-4} preserves invariant~6.
	
	Consider invariant~7.  Since it holds before {\sc Collapse}, we conclude that for all $i \in [1,t-1]$, $f_M(Y_{\leq i}, X_{\leq i+1} \allowbreak \cup I) \geq d_{i+1}$ before step~\ref{step:collapse-2}.  Since $L_t$ is going to be the topmost layer, we only need to show that these inequalities hold after steps~\ref{step:collapse-2}--\ref{step:collapse-4}.  Take any index $i \in [1,t-1]$.  We claim that, before the execution of steps~\ref{step:collapse-2}--\ref{step:collapse-4}, $G_M(Y_{\leq i},X_{\leq i+1} \cup I_{\leq i})$ and $G_M(Y_{\leq i},X_{\leq i+1} \cup I)$ have a common optimal solution $\Pi_{i}$ that is node-disjoint from $\Gamma_t$.  We will prove this claim later.   Assume that the claim is true.  It follows that 
		$f_M(Y_{\leq i},X_{\leq i+1} \cup I_{\leq i}) = |\Pi_{i}| = f_M(Y_{\leq i},X_{\leq i+1} \cup I) \geq d_{i+1}$ before the execution of step 2.  
	Step~\ref{step:collapse-2} sets ${\cal I}:= {\cal I}_{\leq t-1}$, and so after that, $\Pi_{i}$ remains an optimal solution of $G_M(Y_{\leq i},X_{\leq i+1} \cup I)$, and hence 
		$f_M(Y_{\leq i},X_{\leq i+1} \cup I) = |\Pi_{i}| \geq d_{i+1}$.
	Step~\ref{step:collapse-3} changes the matching $M$ by flipping the alternating paths in $\Gamma^+_t$.  Let $M'$ denote the updated matching.  Since $\Gamma_t^+$ is node-disjoint from $\Pi_{i}$ by the claim, flipping the paths in $\Gamma^+_t$ does not affect $\Pi_{i}$, meaning that $\Pi_{i}$ is still a feasible solution of $G_{M'}(Y_{\leq i},X_{\leq i+1} \cup I)$. Thus, 	$f_{M'}(Y_{\leq i},X_{\leq i+1} \cup I)
		\geq |\Pi_{i}| 
		= f_M(Y_{\leq i},X_{\leq i+1} \cup I) 
		\geq d_{i+1}$ 
	after step~3.  In step~\ref{step:collapse-4}, the removal of edges from ${\cal X}_t$ can only affect $f_{M'}(Y_{\leq t-1},X_{\leq t} \cup I)$ and so for $i \in [1,t-2]$, $f_{M'}(Y_{\leq i},X_{\leq i+1} \cup I) \geq d_{i+1}$ after step~\ref{step:collapse-4}.  If $f_{M'}(Y_{\leq t-1},X_{\leq t} \cup I)$ decreases after removing an edge $(p,B)$ from ${\cal X}_t$, that is, 
		$f_{M'}(Y_{\leq t-1}, X_{\leq t} \cup I) 
		= f_{M'}(Y_{\leq t-1}, (X_{\leq t} \cup I) \setminus \{p\}) + 1$, 
	then when we reach step~\ref{step:collapse-4}, Claim~\ref{cl:more-path} will imply that 
		$f_{M'}(Y_{\leq t-1},I \cup \{p\}) = f_{M'}(Y_{\leq t-1},I) + 1$, 
	and so step~\ref{step:collapse-4} will add $p$ to $I$.  Afterwards, $f_{M'}(Y_{\leq t-1},X_{\leq t} \cup I)$ returns to its value prior to the removal of $(p,B)$ from ${\cal X}_t$.  As a result, invariant~7 holds after step~\ref{step:collapse-4}.

	It remains to prove the claim: before executing steps~2--4, for 
		$i \in [1,t-1]$, $G_M(Y_{\leq i},X_{\leq i+1} \cup I_{\leq i})$ 
	and $G_M(Y_{\leq i},X_{\leq i+1} \cup I)$ have a common optimal solution $\Pi_{i}$ that is node-disjoint from $\Gamma_t$.  For convenience, define 
		$\Gamma_{\leq i} = \Gamma_1 \cup \Gamma_2 \ldots \cup \Gamma_i$.
	By the definition of a canonical solution, $\Gamma_{\leq i}$ is an optimal solution of $G_M(Y_{\leq i},I)$.  We augment $\Gamma_{\leq i}$ successively (using Claim~\ref{cl:path-aug}) to obtain $\Pi_{i}$ as an optimal solution of $G_M(Y_{\leq i},X_{\leq i+1} \cup I)$.   It remains to show that $\Pi_{i}$ is an optimal solution of $G_M(Y_{\leq i},X_{\leq i+1} \cup I_{\leq i})$ and that $\Pi_{i}$ is node-disjoint from $\Gamma_t$.  By the definition of canonical solution, $I_{\leq i} = \sink(\Gamma_{\leq i})$.  By claim~\ref{cl:path-aug}, the successive augmentations maintain $I_{\leq i} \subseteq \sink(\Pi_{i})$.  No player in $I_j$ for any $j > i$ can be a sink in $\Pi_{i}$. Otherwise, there would be node-disjoint paths that originate from $Y_{\leq i}$, cover all players in $I_{\leq i}$, and another player in $I_j$ for some $j > i$, contradicting the requirement of a canonical decomposition that $f_M(Y_{\leq i},I) = f_M(Y_{\leq i},I_{\leq i})$.  Hence, $\Pi_{i}$ is also an optimal solution of $G_M(Y_{\leq i},X_{\leq i+1} \cup I_{\leq i})$.
	Next we prove that $\Pi_{i}$ is node-disjoint from $\Gamma_t$.  Recall that $\Pi_{i}$ is obtained from $\Gamma_{\leq i}$ by successive augmentations using Claim~\ref{cl:path-aug}.  Let $\{\Gamma_{\leq i} = \Phi_1, \Phi_2, \ldots, \Phi_k = \Pi_{i}\}$ be the successive sets of node-disjoint paths that are obtained during the successive augmentations.  By Claim~\ref{cl:path-aug}, $I_{\leq i} = \sink(\Gamma_{\leq i}) \subseteq \sink(\Phi_j)$ for all $j \in [1,k]$.  Note that $\Phi_1 = \Gamma_{\leq i}$ is node-disjoint from $\Gamma_t$.  Suppose that $\Phi_{j-1}$ is node-disjoint from $\Gamma_t$ for some $j \in [2,k]$.  We argue that $\Phi_j$ is  node-disjoint from $\Gamma_t$ and then so is $\Phi_k = \Pi_{i}$ by induction.  Recall that $\Phi_{j-1}^+$ is the set of non-trivial paths in $\Phi_{j-1}$. Consider $G_{M\oplus \Phi^+_{j-1}}$.  Let $\gamma$ be the path in $G_{M\oplus \Phi^+_{j-1}}$ that we use to augment $\Phi_{j-1}$ to produce $\Phi_j$.  Since $\Phi_{j-1}$ is node-disjoint from $\Gamma_t$, $G_{M\oplus \Phi^+_{j-1}}$ contains $\Gamma_t$. If $\gamma$ is node-disjoint from $\Gamma_t$, then by Claim~\ref{cl:path-aug}, $\Phi_{j}$ must be node-disjoint from $\Gamma_t$ because the vertex set of $\Phi_j$ is a subset of the vertices of $\Phi_{j-1}$ and $\gamma$.  If $\gamma$ shares a node with some path in $\Gamma_t$, then by switching at that shared node, we have a path $\gamma'$ in $G_{M\oplus \Phi^+_{j-1}}$ that originates from $Y_{\leq i}$ and ends at a sink of $\Gamma_t$ which is in $I_t$.  It means that if we augment $\Phi_{j-1}$ using $\gamma'$, we would obtain a feasible solution of $G_M(Y_{\leq i},X_{\leq i+1} \cup I)$ whose sinks contains all players of $I_{\leq i}$ and some player in $I_t$.  This allows us to extract a feasible solution of $G_M(Y_{\leq i},I)$ whose sinks contains all players of $I_{\leq i}$ and some player in $I_t$.  But then $f_M(Y_{\leq i},I) > |I_{\leq i}|$, a contradiction to the definition of canonical decomposition.
\end{proof}

In the following, we prove two more properties that result from the invariants and the working of the algorithm.  They will be used later in the analysis of the approximation ratio.

\begin{lemma}
	\label{lem:invar}
	Let $(L_1,\ldots,L_\ell)$ be the stack for an arbitrary state of the algorithm. 
	\begin{emromani}
		\item For every $j \in [1,\ell]$, $|X_j| \leq z_j$.
		
		\item For every $j \in [1,\ell]$, $d_j \geq \sum_{i = 1}^j z_i$.
	\end{emromani}
\end{lemma}
\begin{proof}
	Recall that $z_j:= |X_j|$ after {\sc Build} completes the construction of $L_j$ for all $j \in [1,\ell]$.  Afterwards, $z_j$ remains unchanged until $L_j$ is removed from the stack.  $X_j$ may shrink in step~\ref{step:collapse-4} of {\sc Collapse}, but it never grows.  As a result, $z_j \geq |X_j|$ at all times.
	
	We prove (ii) by induction on the chronological order of executing {\sc Build} and {\sc Collapse}. Initially, $L_1$ is the only layer in the stack and $d_1 = z_1 = 0$ because both $\cal I$ and ${\cal X}_1$ are empty sets.  Thus, (ii) holds at the beginning.  Suppose that we build a new layer $L_{k + 1}$ in the build phase.  Let $\cal I'$ be the set of unblocked addable edges newly added to $\cal I$ during the construction of $L_{k + 1}$, and let $I'$ be the set of players covered by ${\cal I}'$. Then we have 
	\allowdisplaybreaks
	\begin{align*}
		d_{k+1} & =  f_M(Y_{\leq k},X_{\leq k+1} \cup {I} \cup {I}') \\
		& \geq  f_M(Y_{\leq k},X_{\leq k} \cup {I}) + |{ X}_{k+1} \cup { I}'| 
		\quad\quad \mbox{($\because$ players in ${ X}_{k+1} \cup { I}'$ are addable)}\\
		& \geq  f_M(Y_{\leq k-1},X_{\leq k} \cup {I}) + | X_{k+1} \cup I'|\\
		& \geq  d_k + |X_{k+1}| + |I'| \quad\quad\quad\quad\quad\quad\quad\quad \mbox{($\because$ $X_{k+1} \cap I' = \emptyset$ and invariant~7 in Table~\ref{tb:invar})} \\
		& \geq  d_k + z_{k+1}  \quad\quad\quad\quad\quad\quad\quad\quad\quad\quad\quad\,\,\, (\because z_{k+1} := |X_{k+1}|) \\
		& \geq  \sum_{i=1}^k z_i + z_{k+1}   \quad\quad\quad\quad\quad\quad\quad\quad\quad\quad \mbox{($\because$ inductive assumption)} \\
		& =  \sum_{i=1}^{k+1} z_i.
	\end{align*}
	Thus, {\sc Build} preserves (ii).  {\sc Collapse} has no effect on (ii) because the values $d_i$'s and $z_i$'s will not change once they are computed until layer $L_i$ is removed.
\end{proof}

\begin{lemma}
	\label{lem:non-collapse}
	Let $(L_1,\ldots,L_\ell)$ be the stack for an arbitrary state of the algorithm.  If $L_i$ is not collapsible for all $i \in [1, \ell-1]$, the following properties are satisfied.
	\begin{emromani}
		\item $|I| < \mu|Y_{\leq \ell-1}|$.
		\item For every $j \in [1,\ell-1]$, $|X_{\leq j+1}| > \sum_{i=1}^{j+1} z_i - \mu |Y_{\leq j}|$.
	\end{emromani}
\end{lemma}
\begin{proof}
	By invariant~6 in Table~\ref{tb:invar}, $f_M(Y_{\leq \ell-1}, I) = |I|$.  Hence, in the canonical decomposition of ${\cal I}$, we have $I_{\leq \ell-1} = I$.  If $|I_{\leq \ell-1}| = |I| \geq \mu |Y_{\leq \ell -1}|$, by the pigeonhole principle, there exists an index $j \in[1,\ell-1]$ such that $|I_j|\geq \mu|Y_j|$.  But then layer $L_j$ is collapsible, a contradiction.
	
	Consider (ii).  Assume to the contrary that there exists $k \in [1,\ell-1]$ such that $|X_{\leq k+1}| \leq \sum_{i=1}^{k+1} z_i - \mu|Y_{\leq k}|$.  Equivalently, $\sum_{i=1}^{k+1} z_i \geq |X_{\leq k+1}| + \mu|Y_{\leq k}|$.  By invariant~7 in Table~\ref{tb:invar} and Lemma~\ref{lem:invar}(ii), 
	$f_M(Y_{\leq k}, X_{\leq k+1} \cup I) 
	\geq d_{k+1} \geq \sum_{i=1}^{k+1} z_i \geq |X_{\leq k+1}| + \mu|Y_{\leq k}|$.  So any optimal solution of $G_M(Y_{\leq k},X_{\leq k+1} \cup I)$ contains at least $\mu|Y_{\leq k}|$ node-disjoint paths from $Y_{\leq k}$ to $I$.  It follows that $f_M(Y_{\leq k},I) \geq \mu|Y_{\leq k}|$.  By the definition of canonical decomposition, $|I_{\leq k}| = f_M(Y_{\leq k},I) \geq \mu|Y_{\leq k}|$. By the pigeonhole principle, there exists some $j \in [1,k]$ such that $|I_j| \geq \mu|Y_j|$.  But then layer $L_j$ is collapsible, a contradiction.
\end{proof}

\section{Polynomial running time and binary search}
\label{sec:poly}

It is clear that each call of {\sc Build} and {\sc Collapse} runs in time polynomial in $\ell$, $m$ and $n$.  So we need to give a bound on $\ell$ (the number of layers in the stack) and the total number of calls of {\sc Build} and {\sc Collapse}.

Lemma~\ref{lem:key} below is the key to obtaining such bounds.  Recall that $\mu$ is the constant we use to determine the collapsibility of layers, i.e., $L_i$ is collapsible if $|I_i| \geq \mu|Y_i|$.  By Lemma~\ref{lem:key}, if no layer is collapsible in the current stack, then the size of the each layer is at least a constant fraction of the total size of all layers below it.   
This guarantees that the algorithm never gets stuck: it can either build a new non-empty layer or collapse some layer.  It will also allow us to obtain a logarithmic bound on the maximum number of layers.  The proof of Lemma~\ref{lem:key} is quite involved and it requires establishing several technical results and the use of competing players.  We defer the proof to Section~\ref{sec:analysis}.

\begin{lemma}
	\label{lem:key}
	Assume that the values $\tau$ and $\lambda$ used by the algorithm satisfy the relations $\tau \leq \tau^*$ and $\lambda = 6+\delta$ for an arbitrary constant $\delta \in (0,1)$.  There exists a constant $\mu \in (0,1)$ dependent on $\delta$ such that for any state $(M,{\cal E},{\cal I},(L_1,\ldots,L_\ell))$ of the algorithm, if $|Y_{i+1}| < \sqrt{\mu}|Y_{\leq i}|$ for some $i \in [1,\ell-1]$, then some layer below $L_{i+1}$ must be collapsible.
\end{lemma}

With Lemma~\ref{lem:key},  an argument similar to that in~\cite[Lemmas~4.10 and 4.11]{AKS17} can show that given a partial allocation, our algorithm can extend it to satisfy one more player in polynomial time.  By repeating the algorithm at most $n$ times, we can extend a maximum matching  of $G$ to an allocation that satisfies all the players.

We sketch the argument in~\cite{AKS17} in the following.  (The counterpart of Lemma~\ref{lem:key} in~\cite{AKS17} uses $\lambda \approx 12.325+\delta$.)   Recall that $P$ is the set of all players.  As $|Y_{\leq i}| \leq |P|$ and $|Y_{\leq i}|$ grows by a factor $1+\sqrt{\mu}$ from layer to layer when no layer is collapsible, the number of layers in the stack is at most $\log_{1+\sqrt{\mu}} |P|$.  Define a signature vector $(s_1,s_2,\ldots,s_\ell,\infty)$ such that $s_i := \lfloor \log_{1/(1-\mu)} |Y_i|\mu^{-i/2}\rfloor$.  

First, it is shown that the signature vector decreases lexicographically after one call of {\sc Build} and one call of {\sc Collapse}, and the coordinates in the signature vector are non-decreasing from left to right~\cite[Lemma 4.10]{AKS17}.  When a new layer is added, the vector gains a new second to rightmost coordinate and so the lexicographical order decreases.  After collapsing the last layer $L_t$ in the collapse phase, $|Y_t|$ drops to $(1-\mu)|Y_t|$ or less.  One can then verify that $s_t$ drops by one or more.  So the signature vector decreases lexicographically.  After  {\sc Collapse}, no layer is collapsible. Lemma~\ref{lem:key} implies immediately that $s_i \geq s_{i-1}$.  

Second, one can verify from the definition that each $s_i$ is bounded by an integer $U$ of value at most $\log |P| \cdot O(\mu^{-3/2}\log\frac{1}{\mu})$. The number of layers is at most $\log_{1+\sqrt{\mu}} |P| \leq U$.  Thus, the sum of coordinates in any signature vector is at most $U^2$.  The number of distinct partitions of the integer $U^2$ is $O(|P|^{O(\mu^{-3/2}\log\frac{1}{\mu})})$, which is an upper bound on the number of distinct signature vectors.    This bounds the number of calls to {\sc Build} and {\sc Collapse}.  Details can be found in~\cite[Lemma 4.11]{AKS17}.  Since each call of {\sc Build} and {\sc Collapse} runs in $\poly(\ell,m,n)$, we conclude that the algorithm runs in $\poly(m,n)$ time.  

The remaining task is to binary search for $\tau^*$.  If we use a value $\tau$ that is at most $\tau^*$, the algorithm terminates in polynomial time with an allocation.  If we use a value $\tau > \tau^*$, there are two possible outcomes.  We may be lucky and always have some collapsible layer below $L_{i+1}$ whenever $|Y_{i+1}| < \sqrt{\mu}|Y_{\leq i}|$ for some $i \in [1,\ell-1]$.  In this case, the algorithm returns in polynomial time an allocation of value at least $\tau/\lambda \geq \tau^*/\lambda$.  The second outcome is that no layer is collapsible at some point, but $|Y_{i+1}| < \sqrt{\mu}|Y_{\leq i}|$ for some $i \in [1,\ell-1]$.  This can be detected in $O(1)$ time by maintaining $|Y_{i+1}|$ and $|Y_{\leq i}|$, which allows us to detect that $\tau > \tau^*$ and halt the algorithm.  Since this is the first violation of the property that some layer below $L_{i+1}$ is collapsible if $|Y_{i+1}| < \sqrt{\mu}|Y_{\leq i}|$ for some $i \in [1,\ell-1]$, the running time before halting is polynomial in $m$ and $n$.  The last allocation returned by the algorithm during the binary search has value at least $\tau^*/\lambda = \tau^*/(6+\delta)$.  We will see in Section~\ref{sec:key} that a smaller $\delta$ requires a smaller $\mu$ and hence a higher running time.

In summary, the initial range for $\tau$ for the binary search process is $[a,b]$, where $a$ is initialized to zero, and $b$ is initialized to be the sum of values of all resources divided by the number of players (rounded down).  In each binary search probe (i.e., $\tau = \lfloor (a+b)/2 \rfloor$), if the algorithm returns an allocation of value at least $\tau/\lambda$, we update $a := \tau+1$ and recurse on $[a,b]$.  Otherwise, we update $b := \tau-1$ and recurse on $[a,b]$.  When the range becomes empty, we stop.  The last allocation returned by the algorithm during the binary search has value at least $\tau^*/\lambda = \tau^*/(6+\delta)$.  This gives the main theorem of this paper as stated below.

\begin{theorem}
	For any fixed constant $\delta \in (0,1)$, there is an algorithm for the restricted max-min fair allocation problem that returns a $(6+\delta)$-approximate solution in time polynomial in the number of players and the number of resources.
\end{theorem}

%% file: src/analysis.tex
\section{Analysis}
\label{sec:analysis}

We will develop lower and upper bounds for the total value of the thin resources in the stack and show that if Lemma~\ref{lem:key} does not hold, the lower bound would exceed the upper bound.  To do this, we need a tool to analyze our aggressive greedy strategy for picking blocked addable thin edges. We introduce this tool in section~\ref{sec:compete} and then prove Lemma~\ref{lem:key} in section~\ref{sec:key}. Recall that we assume $\tau \leq \tau^*$.

\subsection{Competing players}
\label{sec:compete}

In this section, we will show that there is an injective map $\varphi$ from players covered by blocked addable thin edges to players who can access thin resources of large total value.  We call the image of $\varphi$ the competing players.  The next result shows that the target players can be identified via a special maximum matching with respect to an optimal allocation.  Recall that for a maximum matching $M$ of $G$, $\overbar{P}_M$ is the set of players not matched by $M$.

\begin{lemma}
	\label{lem:induced-matching}
	Let {\rm OPT} be an arbitrary optimal allocation, i.e., of value $\tau^*$. There exists a maximum matching $M^*$ of $G$ induced by $\mathrm{OPT}$ such that $M^*$ matches every player who is assigned at least one fat resource in $\mathrm{OPT}$.  Hence, every player in $\overbar{P}_{M^*}$ is assigned only thin resources in {\rm OPT} which are worth a total value of $\tau^*$ or more.
\end{lemma}
\begin{proof}
	We construct a maximum matching of $G$ induced by {\rm OPT} as follows. {\rm OPT} induces a matching $M$ of $G$ that matches all the players who receive at least one fat resource in $OPT$.  $M$ may not be a maximum matching though.  We augment $M$ to a maximum matching $M^*$ using augmenting paths as in basic matching theory~\cite{HK73}.  Augmentation ensures that $M^*$ matches all the players who are matched by $M$.  Hence, $M^*$ is the desired maximum matching.
\end{proof}

Lemma~\ref{lem:inj} below establishes the existence of the injective map $\varphi$ from $X_{\leq \ell}$ to $\overbar{P}_{M^*}$ where $M^*$ is a maximum matching induced by some optimal allocation {\rm OPT}.  Let $\dm_{\varphi}$ and $\im_\varphi$ be the domain and image of $\varphi$, respectively. Consider Lemma~\ref{lem:inj} (i) and (iii).  It would be ideal if $\dm_{\varphi}$ covers the entire $X_{\leq \ell}$. However, for technical reasons, when {\sc Collapse} removes a player from $X_{\leq \ell}$, we may have to remove two players from $\dm_{\varphi}$ in order to maintain other properties of $\varphi$.  Lemma~\ref{lem:inj}(iii) puts a lower bound on the size of the domain of $\varphi$.  

As stated in Lemma~\ref{lem:induced-matching}, players in $\im_\varphi \subseteq \overbar{P}_{M^*}$ are assigned at least $\tau^*|\im_\varphi| \geq \tau|\im_\varphi|$ worth of thin resources in total in {\rm OPT}.  We argue that Lemma~\ref{lem:inj}(ii) implies that a large subset of these thin resources are in the stack.  Consider the time when a player $p$ was added to $X_k$ (and its corresponding addable edge was added to ${\cal X}_k$). The player $p$ was addable at that time, and so was $\varphi(p)$.  Since the algorithm preferred $p$ to $\varphi(p)$, either no addable edge was incident to $\varphi(p)$ or the maximal addable edge identified for $\varphi(p)$ had value no more than the maximal addable edge $e_p$ identified for $p$.  In both cases, at least $\tau-\val(e_p)$ worth of thin resources assigned to $\varphi(p)$ in {\rm OPT} were already in the stack.

To have a good lower bound for the total value of the thin resources in the stack, we also need to look at the players in $\overbar{P}_{M^*}\setminus \im_\varphi$.  Let $\ell$ be the index of the topmost layer in the stack.  Lemma~\ref{lem:inj}(iv) will allow us to prove that roughly $|Y_{\leq \ell-1}| - |X_{\leq \ell}\cup I|$ of players in $\overbar{P}_{M^*}\setminus \im_\varphi$ are still addable after we finish adding edges to ${\cal X}_{\ell}$ during the construction of layer $L_\ell$. However, there are no more addable edges (to be added to ${\cal X}_{\ell}$).  Therefore, each of these addable players can access no more than $\tau/\lambda$ worth of thin resources that are not in the stack.  In other words, at least $\tau - \tau/\lambda$ worth of the thin resources that are assigned to each of them in {\rm OPT} are already in the stack.

\begin{lemma}
\label{lem:inj}  
	Let $M^*$ be a maximum matching of $G$ induced by some optimal allocation.  For any state $(M,{\cal E},{\cal I},(L_1,\ldots,L_\ell))$ of the algorithm, there exists an injection $\varphi$ such that:
	\begin{emromani}
		\item The domain $\dm_\varphi$ and image $\im_\varphi$ of $\varphi$ are subsets of $X_{\leq \ell}$ and $\overbar{P}_{M^*}$, respectively.

		\item For every player $p \in \dm_\varphi$, when $p$ was added to $X_k$ for some $k \in [1,\ell]$, $\varphi(p)$ was also an addable player at that time.

		\item $|\dm_\varphi| \geq 2|X_{\leq \ell}| - \sum_{i = 1}^\ell z_i$.

		\item $f_M(\overbar{P}_M,(\overbar{P}_{M^*} \setminus \im_\varphi) \cup X_{\leq \ell}) 
			= |\overbar{P}_M|$.
\end{emromani}
\end{lemma}
\begin{proof}
	Our proof is by induction on the chronological order of the build and collapse phases. In the base case, $\ell = 1$, $X_1 = \emptyset$, and $z_1 = 0$.  The existence of $\varphi$ is trivial as its domain $D_\varphi \subseteq X_1 = \emptyset$ and image $\im_\varphi =\emptyset$.  Then, (i), (ii), and (iii) are satisfied trivially. As both $\im_\varphi$ and $X_{\leq 1}$ are empty, the left hand side of (iv) becomes $f_M(\overbar{P}_M,\overbar{P}_{M^*})$ which is equal to $|\overbar{P}_M|$ by Claim~\ref{cl:equal-path-num}(i).  We analyze how to update the injection $\varphi$ during the build and collapse phases in order to preserve (i)--(iv).

	\paragraph*{\emph{Build phase}.} Suppose that {\sc Build} begins to construct a new layer $L_\ell$. $X_\ell$ is initialized to be empty.  The value $z_\ell$ is computed only at the completion of $L_\ell$.  However, in this proof, we initialize $z_\ell = 0$, increment $z_\ell$ whenever we add an edge to ${\cal X}_\ell$ (and the corresponding player to $X_\ell$), and show the validity of (i)--(iv) inductively.  This will then imply the validity of (i)--(iv) at the completion of $L_\ell$.

	Since $X_\ell = \emptyset$ and $z_\ell = 0$ initially, properties (i)--(iv) are satisfied by the current $\varphi$ by inductive assumption.

	Step~\ref{step:build-2} of {\sc Build} does not change $X_\ell$, and so $\varphi$ needs no update.

	In step~\ref{step:build-3} of {\sc Build}, when we add an edge to ${\cal X}_\ell$, we need to update $z_\ell$, $\varphi$ and $\dm_\varphi$.  Suppose that a thin edge incident to player $q_1$ is added to ${\cal X}_\ell$.  So $q_1$ is an addable player.  For clarity, we use $X'_\ell$, $z'_\ell$, $\varphi'$, $\dm_{\varphi'}$, and $\im_{\varphi'}$ to denote the updated $X_\ell$, $z_\ell$, $\varphi$, $\dm_\varphi$, and $\im_{\varphi}$, respectively.  Clearly, $X'_\ell = X_\ell \cup \{q_1\}$ and $z'_\ell = z_\ell + 1$.  We set $\dm_{\varphi'} := \dm_{\varphi} \cup \{q_1\}$.  For every player $p \in \dm_{\varphi'} \setminus \{q_1\}$, we set $\varphi'(p) := \varphi(p)$.  We determine $\varphi'(q_1)$ as follows. Recall that for a set $\Pi$ of node-disjoint paths in $G_M$, we use $\Pi^+$ to denote the set of non-trivial paths in $\Pi$.

	Let $\Pi_1$ be an optimal solution of $G_M(Y_{\leq \ell-1}, X_{\leq \ell} \cup I \cup \{q_1\})$.  Player $q_1$ must be a sink of $\Pi_1$ since otherwise we would have 
		$f_M(Y_{\leq \ell-1}, X_{\leq \ell} \cup I \cup \{q_1\}) 
		= f_M(Y_{\leq \ell-1}, X_{\leq \ell} \cup I)$, 
	contradicting the addability of $q_1$.  For the same reason, we have $q_1 \not\in X_{\leq \ell}$.  As $q_1 \in \sink(\Pi_1)$, $q_1$ must be unmatched in the maximum matching 
		$M\oplus \Pi_1^+$, i.e., $q_1 \in \overbar{P}_{M\oplus \Pi^+_1}$.
	Let $\Pi_2$ be an optimal solution of 
		$G_{M\oplus \Pi_1^+}(\overbar{P}_{M\oplus \Pi_1^+},
			(\overbar{P}_{M^*} \setminus \im_\varphi) \cup X_{\leq \ell})$.
	\begin{alignat*}{2}
		|\Pi_2| & =  
			f_{M\oplus \Pi_1^+}(
				\overbar{P}_{M\oplus \Pi_1^+},
				(\overbar{P}_{M^*} \setminus \im_\varphi) \cup X_{\leq \ell}
			)\\
		& =  f_{M}(
				\overbar{P}_{M},
				(\overbar{P}_{M^*} \setminus \im_\varphi) \cup X_{\leq \ell}
			) &\quad\quad \text{(by Claim~\ref{cl:equal-path-num}(ii))} \\
		& =  |\overbar{P}_M| &\quad\quad \text{(by induction assumption)} \\
		& =  |\overbar{P}_{M\oplus \Pi_1^+}|.
	\end{alignat*}
	It follows that every player in $\overbar{P}_{M\oplus \Pi_1^+}$ is a source in $\Pi_2$.  So there exists a path $\pi \in \Pi_2$ that originates from $q_1$.  Let $q_2 = \sink(\pi)$.

	We claim that $q_2 \not\in X_{\leq \ell}$. If $\pi$ is a trivial path, then the claim holds trivially since $q_2 = q_1\notin X_{\leq \ell}$.  Suppose that $\pi$ is a non-trivial path.  Assume, for the sake of contradiction, that $q_2 \in X_{\leq \ell}$.   This allows us to apply Claim~\ref{cl:path-reroute} to $G_M(Y_{\leq \ell-1}, X_{\leq \ell} \cup I \cup \{q_1\})$, it optimal solution $\Pi_1$, and the path $\pi$ (Recall that $q_1 = \sink(\pi)\in \sink(\Pi_1)$). By Claim~\ref{cl:path-reroute}, we can use $\pi$ to convert $\Pi_1$ to an equal-sized set of node-disjoint paths from $Y_{\leq \ell-1}$ to $X_{\leq \ell} \cup I$.  But then 
		$f_M(Y_{\leq \ell-1}, X_{\leq \ell} \cup I) 
		\geq |\Pi_1| = f_M(Y_{\leq \ell-1}, X_{\leq \ell} \cup I \cup \{q_1\})$, 
	a contradiction to the addability of $q_1$.  This proves our claim that $q_2 \not\in X_{\leq \ell}$.

	Observe that $q_2 \in \overbar{P}_{M^*} \setminus \im_\varphi$ because $q_2 \in \sink(\Pi_2) \subseteq (\overbar{P}_{M^*} \setminus \im_{\varphi}) \cup X_{\leq \ell}$ and $q_2 \not\in X_{\leq \ell}$.  This allows us to set $\varphi'(q_1) := q_2$ and keep $\varphi'$ injective.

	Now we show that properties (i)--(iv) are satisfied by $\varphi'$, $z'_\ell$, $\dm_{\varphi'}$, and $X'_\ell$. Properties (i) and (iii) are straightforwardly satisfied.

	By induction assumption, (ii) holds for players in $\dm_{\varphi'} \setminus \{q_1\} = \dm_{\varphi}$.  It remains to check the validity of (ii) for $\varphi'(q_1) = q_2$.  Recall that $\pi$ is a path from $q_1$ to $q_2$ in $G_{M\oplus\Pi_1^+}$. If $\pi$ is a trivial path, then (ii)  holds because $q_2 = q_1$ and $q_1$ is addable. Assume that $\pi$ is non-trivial.  By Claim~\ref{cl:path-reroute}, we can use $\pi$ to convert $\Pi_1$ to an equal-sized set of node-disjoint paths in $G_M$ from $Y_{\leq \ell-1}$ to $ X_{\leq \ell}\cup I \cup \{q_2\}$.  Thus, 
		$f_M(Y_{\leq \ell-1}, X_{\leq \ell}\cup I \cup \{q_2\}) \geq |\Pi_1| 
		= f_M(Y_{\leq \ell-1}, X_{\leq \ell}\cup I \cup \{q_1\})$, 
	which is equal to $f_M(Y_{\leq \ell-1}, X_{\leq \ell} \cup I) + 1$ as $q_1$ is addable.  Therefore, $q_2$ is also an addable player at the time when ${\cal X}_\ell$ gains a thin edge incident to $q_1$. Then (ii) holds for $\varphi'(q_1) = q_2$.

	Consider (iv).  If $\pi$ is a trivial path, i.e., $q_1 = q_2$, then (iv) holds because
	\[
		(\overbar{P}_{M^*} \setminus \im_\varphi) \cup X_{\leq \ell} 
		\subseteq (\overbar{P}_{M^*} \setminus (\im_\varphi \cup \{q_2\})) \cup X_{\leq \ell}\cup \{q_1\} 
		= (\overbar{P}_{M^*} \setminus \im_{\varphi'}) \cup X'_{\leq \ell} .
	\]
	Suppose that $\pi$ is non-trivial.  Recall that $\Pi_2$ is an optimal solution of 
		$G_{M\oplus \Pi^+_1}(
			\overbar{P}_{M\oplus \Pi^+_1},
			(\overbar{P}_{M^*} \setminus \im_\varphi) \cup X_{\leq \ell}
		)$, 
	and that $|\Pi_2| = |\overbar{P}_{M\oplus \Pi_1^+}|$.  Take the maximum matching $M\oplus \Pi^+_1$ of $G$ and flip the paths in $\Pi^+_2 \setminus\{\pi\}$ in $G$.  This produces another maximum matching $M' = (M\oplus \Pi^+_1)\oplus (\Pi^+_2\setminus \{\pi\})$.  All $|\overbar{P}_{M\oplus \Pi_1^+}|$ sinks of $\Pi_2$, except for $q_2$, are unmatched in $M'$.  Player $q_1$ is also unmatched in $M'$.  There are equally many unmatched players in $M'$ and $M\oplus \Pi_1^+$ as both are maximum matchings of $G$.  This implies that 
		$(\sink(\Pi_2) \setminus \{q_2\}) \cup \{q_1\}$ 
	is exactly $\overbar{P}_{M'}$.  Since 
		$\sink(\Pi_2) \subseteq (\overbar{P}_{M^*} \setminus \im_\varphi) \cup X_{\leq \ell}$, 
	we conclude that
	\begin{align*}
		\overbar{P}_{M'} 
		&\subseteq  \left(
		\left((\overbar{P}_{M^*} \setminus \im_\varphi) \cup X_{\leq \ell}\right)
		\setminus \{q_2\}\right) \cup \{q_1\}\\
		&\subseteq  (\overbar{P}_{M^*} \setminus (\im_\varphi \cup \{q_2\})) 
		\cup X_{\leq \ell}\cup \{q_1\}\\
		& =  (\overbar{P}_{M^*} \setminus \im_{\varphi'}) \cup X'_{\leq \ell}.
	\end{align*}
	By the above subset relation and Claim~\ref{cl:equal-path-num}(i),
	\begin{equation*}
		|\overbar{P}_M| 
		\geq f_M(\overbar{P}_{M},(\overbar{P}_{M^*} \setminus \im_{\varphi'}) \cup X'_{\leq \ell})
		\geq f_M(\overbar{P}_M,\overbar{P}_{M'})
		= |\overbar{P}_{M}|.
	\end{equation*}
	Hence, (iv) holds.
	
	Clearly, steps~\ref{step:build-4}--\ref{step:build-6} of {\sc Build} do not affect $\varphi$.

	\paragraph*{\emph{Collapse phase}.} Suppose that we are going to collapse the layer $L_t$.  Since we will set $\ell := t$ at the end of collapsing $L_t$, we only need to prove (i)---(iv) with $\ell$ substituted by $t$. 

	Clearly, step~\ref{step:collapse-1} of {\sc Collapse} has no effect on $\varphi$. 
	
	Consider step~\ref{step:collapse-2} of {\sc Collapse}.  Go back to the last time when $L_t$ was either created by {\sc Build} as the topmost layer or made by {\sc Collapse} as the topmost layer.  By the inductive assumption, there was an injection $\varphi''$ at that time that satisfies (i)--(iv).  We set $\varphi := \varphi''$, $\dm_{\varphi}: = \dm_{\varphi''}$, and $\im_{\varphi} := \im_{\varphi''}$.

	In step~\ref{step:collapse-3} of {\sc Collapse}, the maximum matching $M$ may change, so only (iv) is affected.  Nonetheless, by Claim~\ref{cl:equal-path-num}(ii), the value of $f_M(\overbar{P}_M, (\overbar{P}_{M^*} \setminus \im_\varphi) \cup X_\ell)$ remains the same after updating $M$.  So (iv) is satisfied afterwards.

	In step~4 of {\sc Collapse}, we may remove some edges from ${\cal X}_t$ and add some edges to ${\cal I}$.  Adding edges to ${\cal I}$ does not affect $\varphi$. We need to update $\varphi$ when an edge is removed from ${\cal X}_t$.  Suppose that we are going to remove from ${\cal X}_{t}$ an edge that is incident to a player $q_1$.  Let $X'_{\leq t}$, $\varphi'$, $\dm_{\varphi'}$, and $\im_{\varphi'}$ denote the updated $X_{\leq t}$, $\varphi$, $\dm_{\varphi}$, and $\im_{\varphi}$, respectively.  Note that $X'_{\leq t} = X_{\leq t} \setminus \{q_1\}$.  We show how to define $\varphi'$, $\dm_{\varphi'}$, and $\im_{\varphi'}$ appropriately.  Recall that $z_{t}$ was defined in the last construction of the layer $L_t$ during the build phase, and it has remained fixed despite possible changes to ${\cal X}_t$ since then.

	Consider (iv). If (iv) is not affected by replacing $X_{\leq \ell}$ with $X'_{\leq \ell}$, that is, 
		$f_M(\overbar{P}_M,(\overbar{P}_{M^*} \setminus \im_{\varphi}) \cup X'_{\leq t})
		= |\overbar{P}_M|$,
	then we simply set $\dm_{\varphi'} := \dm_{\varphi} \setminus \{q_1\}$ and $\varphi'(p) := \varphi(p)$ for all $p \in \dm_{\varphi'}$.  Since $\im_{\varphi'} \subseteq \im_{\varphi}$, $\varphi'$ satisfies (iv), i.e., $f_M(\overbar{P}_M,(\overbar{P}_{M^*} \setminus \im_{\varphi'}) \cup X'_{\leq t})= |\overbar{P}_M|$. Suppose that property (iv) is affected by replacing $X_{\leq \ell}$ with $X'_{\leq \ell}$, and as a consequence, 
	\begin{equation}
		f_M(\overbar{P}_M,(\overbar{P}_{M^*} \setminus \im_{\varphi}) \cup X'_{\leq t}) 
		= |\overbar{P}_M| - 1.
		\label{eq:inj-collapse-1}
	\end{equation}
	Since $f_M(\overbar{P}_M, \overbar{P}_{M^*}) = |\overbar{P}_M|$, we have  
	\begin{equation}
		f_M(\overbar{P}_M,\overbar{P}_{M^*} \cup X'_{\leq t})= |\overbar{P}_M|.
		\label{eq:inj-collapse-2}
	\end{equation}
	Comparing equations~(\ref{eq:inj-collapse-1}) and (\ref{eq:inj-collapse-2}), we conclude that there exists a player $q_2 \in \dm_{\varphi}$ such that 
	\[
		f_M(\overbar{P}_M,
		(\overbar{P}_{M^*} \setminus (\im_{\varphi}\setminus\{\varphi(q_2)\})) \cup X'_{\leq t})
		= |\overbar{P}_M|.
		\label{eq:inj-collapse-3}
	\]
	We set $\dm_{\varphi'} := \dm_{\varphi}\setminus\{q_1,q_2\}$, and $\varphi'(p) := \varphi(p)$ for all $p \in \dm_{\varphi'}$. Then (iv) is satisfied by $\varphi'$.
	
	Irrespective of which definition of $\varphi'$ above is used,  properties (i) and (ii) trivially hold.  Property (iii) holds because the left hand side decreases by at most $2$ and the right hand side decreases by exactly $2$.
\end{proof}

\subsection{Proof of Lemma~\ref{lem:key}}
\label{sec:key}

Suppose, for the sake of contradiction, that there exists an index $k \in [1,\ell-1]$ such that $|Y_{k+1}| < \sqrt{\mu}\,|Y_{\leq k}|$ but no layer below $L_{k+1}$ is collapsible.  Without loss of generality, we assume that $k$ is the smallest such index. Therefore, $|Y_{i+1}| \geq \sqrt{\mu}\,|Y_{\leq i}|$ for every $i \in [1,k-1]$.

Consider the moment immediately after the last construction of the $(k+1)$-th layer in the build phase. Let $(M', {\cal E}',{\cal I}',(L'_1,\ldots,L'_{k+1}))$ be the state of the algorithm at that moment, where $L'_{i} = ({\cal X}'_{i}, {\cal Y}'_{i}, d'_{k+1}, z'_{i})$.

We will derive a few inequalities that hold given the existence of $k$, and then obtain a contradiction by showing that the system of these inequalities is infeasible.

We first define some notation that will be used in the proof. Let $M^*$ be a maximum matching induced by an optimal allocation {\rm OPT}.  Let $\varphi'$ and $\dm_{\varphi'}$ be the injection and its domain as defined in Lemma~\ref{lem:inj} with respect to $M^*$ and the state $(M', {\cal E}',{\cal I}',(L'_1,\ldots,L'_{k+1}))$. Define 
\[
\forall\, p \in \dm_{\varphi'}, \, w_p := \val(B), \, \mbox{where $(p,B)$ is the thin edge in ${\cal X}'_{\leq k+1}$ that is incident to $p$}.  
\]
It is well defined because $\dm_{\varphi'} \subseteq X'_{\leq k+1}$ by Lemma~\ref{lem:inj}(i) and no player is incident to two edges in ${\cal X}'_{\leq k+1}$ by invariant~1 in Table~\ref{tb:invar}.  Also by invariant~1 in Table~\ref{tb:invar},
\[
\forall \, p \in \dm_{\varphi'}, \, w_p \in [\tau/\lambda,\tau + \tau/\lambda].
\]
Given the above definition, we already have two easy inequalities. Recall that given a set $\cal S$ of thin edges, $\val({\cal S})$ is the total value of the thin resources covered by $\cal S$.
\[
 	\val({\cal X}'_{\leq k+1} \cup {\cal Y}'_{\leq k})
 	\geq \val({\cal X}'_{\leq k+1}) 
 	\geq \sum_{p \in \dm_{\varphi'}} w_p
\]
\[
	\frac{\tau}{\lambda}|\dm_{\varphi'}| 
	\leq \sum_{p \in \dm_{\varphi'}} w_p 
	\leq \left(\tau + \frac{\tau}{\lambda}\right)|\dm_{\varphi'}|
\]
We list three more inequalities in the claims below.  Their proofs are deferred to Sections~\ref{sec:cl:ineq}, \ref{sec:cl:upper} and~\ref{sec:cl:lower}, respectively.

\begin{claim}
	\label{cl:ineq}
	$|\dm_{\varphi'}| \leq |Y'_{\leq k}|$.
\end{claim}

\begin{claim}
	\label{cl:upper}
	$\val({\cal X}'_{\leq k+1} \cup {\cal Y}'_{\leq k}) \leq \frac{\tau}{\lambda}|\dm_{\varphi'}| +
		\frac{2\tau}{\lambda}|Y'_{\leq k}| + \frac{\delta_1\tau}{\lambda}|Y'_{\leq k}|$, where $\delta_1 = \lambda\mu + 2\mu + 2\sqrt{\mu}$.
\end{claim}

\begin{claim}
	\label{cl:lower}
	$\val({\cal X}'_{\leq k+1}\cup {\cal Y}'_{\leq k}) \geq (\tau - \frac{\tau}{\lambda})(|Y'_{\leq k}|- |\dm_{\varphi'}|) + \sum_{p\in \dm_{\varphi'}}(\tau - w_p) 
	- \frac{\delta_2\tau}{\lambda}|Y'_{\leq k}|$, where $\delta_2 = 2\lambda\mu + 2\lambda\sqrt{\mu} + 6\sqrt{\mu}$.
\end{claim}

Putting all the five inequalities together gives the following system:
\allowdisplaybreaks
\begin{align*}
	\val({\cal X}'_{\leq k+1} \cup {\cal Y}'_{\leq k}) & \geq  \sum_{p \in \dm_{\varphi'}} w_p,\\
	\frac{\tau}{\lambda}|\dm_{\varphi'}| & \leq  \sum_{p \in \dm_{\varphi'}} w_p \,\,\,\, \leq \,\,\,\, (\tau + \tau/\lambda)|\dm_{\varphi'}|,\\
	|\dm_{\varphi'}| & \leq  |Y'_{\leq k}|,\\
	\val({\cal X}'_{\leq k+1} \cup {\cal Y}'_{\leq k}) & \leq  \frac{\tau}{\lambda}|\dm_{\varphi'}| +
	\frac{2\tau}{\lambda}|Y'_{\leq k}| + \frac{\delta_1\tau}{\lambda}|Y'_{\leq k}|,\\
	\val({\cal X}'_{\leq k+1}\cup {\cal Y}'_{\leq k}) & \geq  (\tau - \tau/\lambda)(|Y'_{\leq k}|- |\dm_{\varphi'}|) + \sum_{p\in \dm_{\varphi'}}(\tau - w_p) - \frac{\delta_2\tau}{\lambda}|Y'_{\leq k}|.
\end{align*}
Divide the above system by $\frac{\tau}{\lambda}|\dm_{\varphi'}|$.  To simplify the system, define the variables 
	$B_1 := \val({\cal X}'_{\leq k+1} \cup {\cal Y}'_{\leq k}) / (\frac{\tau}{\lambda}|\dm_{\varphi'}|)$, 
$B_2 := |Y'_{\leq k}|/|\dm_{\varphi'}|$, and 
	$B_3 := (\sum_{p \in \dm_{\varphi'}} w_p)/(\frac{\tau}{\lambda}|\dm_{\varphi'}|)$.
Then we can write the above system equivalently as follows:
\begin{align*}
	B_1 \,\,& \geq \,\,  B_3,\\
	1 \,\,& \leq  \,\, B_3 \,\,\,\, \leq \,\,\,\, \lambda + 1,\\
	1 \,\,& \leq  \,\, B_2,\\
	B_1 \,\,& \leq \,\, 1 + 2B_2 + \delta_1B_2,\\
	B_1 \,\,& \geq \,\, (\lambda - 1)(B_2 -  1) + \lambda - B_3 - \delta_2B_2.
\end{align*}
The first, fourth, and fifth inequalities give
\begin{align}
	&2(1 + 2B_2 + \delta_1B_2) \geq B_1 + B_3 \geq (\lambda - 1)(B_2 -  1) + \lambda - \delta_2B_2  \nonumber \\
 	\Rightarrow&  2 + (4 + 2\delta_1)B_2 \geq (\lambda-1-\delta_2)B_2 + 1 \nonumber \\
 	\Rightarrow & (\lambda - 5 - 2\delta_1 - \delta_2)B_2 \leq  1
	\label{eq:key-1}
\end{align}
When $\mu$ tends to zero, both $\delta_1$ and $\delta_2$ tend to 0. Given that $\lambda = 6 + \delta$ for an arbitrary constant $\delta \in (0,1)$, when $\mu$ is sufficiently small, 
	$\lambda - 5 - 2\delta_1 - \delta_2 > 1$. 
Since $B_2 \geq 1$ by the third inequality, we obtain 
\begin{equation}
	(\lambda - 5 - 2\delta_1 - \delta_2)B_2 >  1.
	\label{eq:key-2}
\end{equation}
This is impossible because \eqref{eq:key-1} and \eqref{eq:key-2} contradict each other.  Hence, Lemma~\ref{lem:key} is true.

\subsection{Proofs of Claims~\ref{cl:ineq}--\ref{cl:lower}}

Recall a few things. $(M,{\cal E},{\cal I},(L_1,\ldots,L_\ell))$ is the current state of the algorithm.  The index $k$ is the smallest index such that $|Y_{k+1}| < \sqrt{\mu}|Y_{\leq k}|$ and no layer below $L_{k+1}$ is collapsible.  Hence, $|Y_{i+1}| \geq \sqrt{\mu}\,|Y_{\leq i}|$ for every $i \in [1,k-1]$.  $(M',{\cal E}',{\cal I}',(L'_1,\cdots,L'_{k+1}))$ is the state of the algorithm immediately after the last construction of the $(k+1)$-th layer in the build phase.  $\varphi'$, $\dm_{\varphi'}$ and $\im_{\varphi'}$ are the injection, its domain, and its image defined in Lemma~\ref{lem:inj} with respect to a maximum matching $M^*$ induced by an optimal allocation {\rm OPT} and the state $(M',{\cal E}',{\cal I}',(L'_1,\cdots,L'_{k+1}))$.

Before proving Claims~\ref{cl:ineq}--\ref{cl:lower}, we derive two more claims.
\begin{claim}
\label{cl:same-layer}
For $i \in [1, k]$, $L'_i = L_i$. That is, none of layers $L'_1, \ldots, L'_k$ have ever been collapsed since the last construction of the $(k+1)$-th layer in the build phase. For layers $L'_{k+1} = ({\cal X}'_{k+1}, {\cal Y}'_{k+1}, d'_{k+1}, z'_{k+1})$ and $L_{k+1} = ({\cal X}_{k+1}, {\cal Y}_{k+1}, d_{k+1}, z_{k+1})$, we have ${\cal X}_{k+1}\subseteq {\cal X}'_{k+1}$, ${\cal Y}_{k+1} \subseteq {\cal Y}'_{k+1}$, $d_{k+1} = d'_{k+1}$, and $z_{k+1} = z'_{k+1}$.
\end{claim}
If some of $L'_1, \ldots, L'_k$ has ever been collapsed, then $L'_{k+1}$ would be removed by {\sc Collapsed} and it could not be the \emph{last} $(k+1)$-th layer constructed so far.  As to layer $L'_{k+1}$, since its construction so far, it has never been destructed, but it may have been shrunk by {\sc Collapse}, and the resulting layer is $L_{k+1}$.

\begin{claim}
	\label{cl:large-domain}
	$|\dm_{\varphi'}| > |X'_{\leq k+1}| - \mu|Y'_{\leq k}|$.
\end{claim}
By Claim~\ref{cl:same-layer}, none of layers $L'_1, \ldots, L'_k$ is collapsible with respect to $(M',{\cal E}',{\cal I}',(L'_1,\cdots,L'_{k+1}))$. By Lemma~\ref{lem:non-collapse}(ii), $|X'_{\leq k+1}| > \sum_{i=1}^{k+1} z_i - \mu |Y_{\leq k}|$. By Lemma~\ref{lem:inj}(iii), $|\dm_\varphi'| \geq 2|X'_{\leq k+1}| - \sum_{i = 1}^{k+1} z_i$. Combining the two inequalities proves Claim~\ref{cl:large-domain}.

\subsubsection{Proof of Claim~\ref{cl:ineq}}
\label{sec:cl:ineq}

By invariant~7 in Table~\ref{tb:invar}, $|Y'_{\leq k}| \geq f_{M'}(Y'_{\leq k},X'_{\leq k+1} \cup I') \geq d_{k+1}$.  By Lemma~\ref{lem:invar}, $d_{k+1} \geq \sum_{i=1}^{k+1}z_{k+1} \geq |X'_{\leq k+1}|$ which is at least $|\dm_{\varphi'}|$ as $\dm_{\varphi'} \subseteq X'_{\leq k+1}$.  Putting the above together shows that $|\dm_{\varphi'}| \leq |Y'_{\leq k}|$ as stated in Claim~\ref{cl:ineq}. 

\subsubsection{Proof of Claim~\ref{cl:upper}}
\label{sec:cl:upper}

We derive upper bounds for $\val({\cal X}'_{\leq k}\cup {\cal Y}'_{\leq k})$ and $\val({\cal X}'_{k+1})$ separately.  Combining them proves the claim. 

Every edge $e \in {\cal Y}'_{\leq k}$ belongs to some partial allocation.  By the definition of a partial allocation, every thin edge in it is minimal, i.e., taking away a thin resource from $e$ puts its value below $\tau/\lambda$, which implies that $\val(e) \leq 2\tau/\lambda$. 
Every edge in ${\cal X}'_{\leq k}$ is blocked, so it has less than $\tau/\lambda$ worth of resources that are not covered by edges in ${\cal Y}'_{\leq k}$. We conclude that 
\begin{equation}
	\val({\cal X}'_{\leq k}\cup {\cal Y}'_{\leq k}) < \frac{\tau}{\lambda}|X'_{\leq k}| + \frac{2\tau}{\lambda}|Y'_{\leq k}|.  
	\label{eq:upper-1}
\end{equation}
We bound $\val({\cal X}'_{k+1})$ next.  By Claim~\ref{cl:same-layer}, 
\begin{alignat*}{2}
|	X'_{k+1} \setminus X_{k+1}| & =  |X'_{\leq k+1}| - |X_{\leq k+1}| \\
	& \leq  \sum_{i=1}^{k+1} z_i - |X_{\leq k+1}| 
		& \quad\quad \text{(by Lemma~\ref{lem:invar}(i))} \\
	& <  \sum_{i=1}^{k+1} z_i - (\sum_{i=1}^{k+1} z_i - \mu|Y_{\leq k}|) 
		& \quad\quad \text{(by Lemma~\ref{lem:non-collapse}(ii))} \\
	& = \mu|Y_{\leq k}| \\
	& = \mu|Y'_{\leq k}|.
\end{alignat*}
Every edge in ${\cal X}'_{k+1}$ has value at most $\tau + \tau/\lambda$ by invariant~1 in Table~\ref{tb:invar}.  Therefore,
\begin{align}
	\val({\cal X}'_{k+1}) 
	& \leq  \val({\cal X}_{k+1}) 
		+ (\tau + \tau/\lambda)(|{\cal X}'_{k+1} \setminus {\cal X}_{k+1}|) \nonumber \\
	& <  \val({\cal X}_{k+1}) + \mu(\tau + \tau/\lambda)|Y'_{\leq k}|.
\label{eq:upper-2}
\end{align}
A reasoning analogous to that behind \eqref{eq:upper-1} gives
\begin{align}
	\val({\cal X}_{k+1}) 
	& <  \frac{\tau}{\lambda}|X_{k+1}| + \frac{2\tau}{\lambda}|Y_{k+1}| \nonumber \\
	& <  \frac{\tau}{\lambda}|X_{k+1}| + \frac{2\tau\sqrt{\mu}}{\lambda}|Y_{\leq k}|
	\quad\quad \text{(by assumption that $|Y_{k+1}| < \sqrt{\mu}|Y_{\leq k}|$)} \nonumber \\
	& =  \frac{\tau}{\lambda}|X_{k+1}| + \frac{2\tau\sqrt{\mu}}{\lambda}|Y'_{\leq k}|. \quad\quad \text{(by Claim~\ref{cl:same-layer})}
\label{eq:upper-3}
\end{align} 
Combining inequalities (\ref{eq:upper-2}) and (\ref{eq:upper-3}) yields
\begin{align}
	\val({\cal X}'_{k+1}) 
	& <  \frac{\tau}{\lambda}|X_{k+1}| + \frac{2\tau\sqrt{\mu}}{\lambda}|Y'_{\leq k}| +\mu(\tau + \tau/\lambda)|Y'_{\leq k}| \nonumber \\
	& \leq  \frac{\tau}{\lambda}|X'_{k+1}| + (\lambda\mu + \mu + 2\sqrt{\mu})\frac{\tau}{\lambda}|Y'_{\leq k}|.  \label{eq:upper-4}
\end{align}
The last inequality is due to ${\cal X}_{k+1}\subseteq {\cal X}'_{k+1}$. 

Putting \eqref{eq:upper-1} and \eqref{eq:upper-4} together, we obtain
\[
	\val({\cal X}'_{\leq k+1} \cup {\cal Y}'_{\leq k}) 
	< \frac{\tau}{\lambda}|X'_{\leq k+1}| 
		+ \frac{2\tau}{\lambda}|Y'_{\leq k}| 
		+ (\lambda\mu + \mu + 2\sqrt{\mu})\frac{\tau}{\lambda}|Y'_{\leq k}|.
\]
By Claim~\ref{cl:large-domain}, $|\dm_{\varphi'}| > |X'_{\leq k+1}| - \mu|Y'_{\leq k}|$. Substituting it into the above inequality gives the following inequality as stated in Claim~\ref{cl:upper}.
\[
	\val({\cal X}'_{\leq k+1} \cup {\cal Y}'_{\leq k}) 
	< \frac{\tau}{\lambda}|\dm_{\varphi'}| 
		+ \frac{2\tau}{\lambda}|Y'_{\leq k}| 
		+ (\lambda\mu + 2\mu + 2\sqrt{\mu})\frac{\tau}{\lambda}|Y'_{\leq k}|.
\]

\subsubsection{Proof of Claim~\ref{cl:lower}}
\label{sec:cl:lower}
Recall the things stated at the beginning of Section~\ref{sec:analysis}.

For each player $p$, let $A_p$ denote the set of resources assigned to $p$ in {\rm OPT}.  By Lemma~\ref{lem:induced-matching}, for any $p \in \overbar{P}_{M^*}$, $A_p$ consists of thin resources only and $\val(A_p) \geq \tau$.  Since {\rm OPT} is an allocation, $A_p \cap A_q = \emptyset$ for any distinct players $p$ and $q$.  For any subset $S$ of players, we define $A_S := \bigcup_{p \in S}A_p$.

Let $B$ denote the set of resources covered by the edges in ${\cal X}'_{\leq k+1} \cup {\cal Y}'_{\leq k}$.  To derive a lower bound of $\val(B)$, it suffices to focus on the value of a subset of $B$.  In particular, we are interested in resources in $B$ that are allocated to $\overbar{P}_{M^*}$ in {\rm OPT}, i.e., $B \cap A_{\overbar{P}_{M^*}}$. $B \cap A_{\overbar{P}_{M^*}}$ can be divided into two disjoint subsets: those in $A_{\im_{\varphi'}}$ and those in $A_{\overbar{P}_{M^*} \setminus \im_{\varphi'}}$.  We consider them separately in the following analysis.

\paragraph*{Resources in $B \cap A_{\overbar{P}_{M^*} \setminus \im_{\varphi'}}$.}
Lemma~\ref{lem:inj}(iv) implies that 
\begin{equation}
	f_{M'}(\overbar{P}_{M'},(\overbar{P}_{M^*}
	\setminus \im_{\varphi'}) \cup X'_{\leq k+1} \cup I') = |\overbar{P}_{M'}|.
	\label{eq:lower-1}
\end{equation}
We have ${\cal Y}'_{i} \subseteq {\cal E}'$ for $i\in [2, k]$ by invariant~4 in Table~\ref{tb:invar}, and ${\cal Y}_1 = \{(p_0, \emptyset)\}$ where $p_0$ is the player we want to satisfy. Therefore, the players in $Y'_{\leq k}$ are not matched by $M'$, i.e., $Y'_{\leq k} \subseteq \overbar{P}_{M'}$.  It follows from \eqref{eq:lower-1} that
\begin{equation}
	f_{M'}(Y'_{\leq k},(\overbar{P}_{M^*}
	\setminus \im_{\varphi'}) \cup X'_{\leq k+1} \cup I') = |Y'_{\leq k}|. 
\label{eq:lower-2}
\end{equation}
We also have
\begin{equation}
	f_{M'}(Y'_{\leq k}, X'_{\leq k+1} \cup I') \leq |X'_{\leq k+1}| + |I'|.
\label{eq:lower-3}
\end{equation}
Comparing \eqref{eq:lower-2} and \eqref{eq:lower-3}, we conclude that, during the construction 
of $L'_{k+1}$, after we finish adding edges to ${\cal I}'$ and ${\cal X}'_{k+1}$ in steps~\ref{step:build-2} and~\ref{step:build-3} of {\sc Build}, at least $|Y'_{\leq k}|- |X'_{\leq k+1}| - |I'|$ players in $\overbar{P}_{M^*} \setminus \im_{\varphi'}$ are still addable.   Let $S$ denote this subset of addable players in $\overbar{P}_{M^*} \setminus \im_{\varphi'}$.  Still, no more edge is added to ${\cal X}'_{\leq k+1} \cup {\cal I}'$ in steps~2 and~3 of {\sc Build}.  The reason must be that the players in $S$ do not have addable edges.   In other words, for each player $p \in S$, since $\val(A_p)\geq \tau$, at least $\tau - \tau/\lambda$ worth of the thin resources in $A_p$ must already be covered by ${\cal X}'_{\leq k+1}\cup {\cal Y}'_{\leq k} \cup {\cal I}'$.  Therefore, the value of resources in $A_S$ that are covered by 
	${\cal X}'_{\leq k+1}\cup {\cal Y}'_{\leq k} \cup {\cal I}'$
is at least
\[
	(\tau - \tau/\lambda)(|Y'_{\leq k}|- |X'_{\leq k+1}| - |I'|).
\]
Subtracting the contribution of ${\cal I}'$, we obtain
\begin{align}
	\val(B \cap A_{\overbar{P}_{M^*} \setminus \im_{\varphi'}}) 
	& \geq  \val(B \cap A_S) \nonumber \\ 
	& \geq  (\tau - \tau/\lambda)(|Y'_{\leq k}|- |X'_{\leq k+1}| - |I'|) - \val({\cal I'}).  \label{eq:lower-4}
\end{align}
By Claim~\ref{cl:same-layer}, $L'_i$ is not collapsible for $i \in [1,k]$. Then by Lemma~\ref{lem:non-collapse}(i), $|I'| < \mu|Y'_{\leq k}|$. By invariant~1 in Table~\ref{tb:invar}, edges in ${\cal I}'$ have values at most $2\tau/\lambda$, so $\val({\cal I}') \leq \frac{2\tau}{\lambda}|I'|$.  This allow us to modify \eqref{eq:lower-4} and get
\begin{align}
	\val(B \cap A_{\overbar{P}_{M^*} \setminus \im_{\varphi'}}) 
	& \geq 
	(\tau - \tau/\lambda)(|Y'_{\leq k}|- |X'_{\leq k+1}|) -  (\tau + \tau/\lambda)|I'| \nonumber \\
	& > 
	(\tau - \tau/\lambda)(|Y'_{\leq k}|- |X'_{\leq k+1}|) - (\tau + \tau/\lambda)\mu|Y'_{\leq k}|.  \label{eq:lower-5}
\end{align}
By Claim~\ref{cl:large-domain}, $|\dm_{\varphi'}| > |X'_{\leq k+1}| - \mu|Y'_{\leq k}|$.  Substituting this inequality into \eqref{eq:lower-5} gives the final inequality that we want in this case.
\begin{equation}
	\val(B \cap A_{\overbar{P}_{M^*} \setminus \im_{\varphi'}}) 
	> (\tau - \tau/\lambda)(|Y'_{\leq k}|- |\dm_{\varphi'}|) - 2\mu\tau|Y'_{\leq k}|.
\label{eq:lower-key-1}
\end{equation}

\paragraph*{Resources in $B \cap A_{\im_{\varphi'}}$.}
As $\dm_{\varphi'} \subseteq X'_{\leq k+1}$ by Lemma~\ref{lem:inj}(i), every player in $\dm_{\varphi'}$ belongs to $X'_j$ for some $j \in [1,k+1]$.  For every $j \in [1,k]$, let $t_j$ denote the time immediately after the last construction of the $j$-th layer in the building phase prior to the creation of $L'_{k+1}$. Similarly let $t_{k+1}$ be the time immediately after the construction of $L'_{k+1}$. One can see that $t_1 < t_2 <\cdots < t_{k+1}$.  For each $t_j$ and $c \in [1, j]$, we use ${\cal X}^{t_j}_{c}$ and ${\cal Y}^{t_j}_{c}$ to denote the set of addable edges and the set of blocking edges in the $c$-th layer at time $t_j$. Similarly, we define ${\cal I}^{t_j}$ to be the set of unblocked addable edges at time $t_j$.

For reasons similar to those underlying Claim~\ref{cl:same-layer}, from time $t_j$ to time $t_{k+1}$, none of the layers below the $j$-th layer have ever been collapsed. As a result, we have ${\cal X}^{t_j}_c = {\cal X}'_c$ and ${\cal Y}^{t_j}_c = {\cal Y}'_c$ for $c \in [1, j-1]$. The $j$-th layer may have been collapsed several times, so ${\cal X}'_j \subseteq {\cal X}^{t_j}_{j}$. We also have $z'_j = |{\cal X}^{t_j}_j|$ because the $z'_j$ is unchanged from time $t_j$ to time $t_{k+1}$.

For every $q \in \dm_{\varphi'}$, let $e_q$ denote the edge in ${\cal X}'_{\leq k+1}$ that is incident $q$.  Recall that $w_{q} = \val(e_q)$ by definition.  Let $q^* = \varphi'(q)$.
\begin{quote}
	\begin{claim}
		\label{cl:lower-1}
		Suppose that $q \in \dm_{\varphi'} \cap {X}'_j$ for some $j \in [2,k+1]$.  The total value of resources shared by ${\cal X}^{t_j}_{\leq j}\cup {\cal Y}^{t_j}_{\leq j-1}\cup {\cal I}^{t_j}$ and $A_{q^*}$ is at least $\tau - w_{q}$.
	\end{claim}
	\begin{proof}
		If $w_{q} \geq \tau$, the correctness is trivial.  Assume that $w_{q} < \tau$. Recall that $e_q \in {\cal X}'_j \subseteq {\cal X}^{t_j}_j$.  Consider the moment when $q$ is picked by the algorithm and $e_q$ is added to ${\cal X}^{t_j}_{j}$.  By Lemma~\ref{lem:inj}(i) and (ii), $q^*$ is also an addable player.  Since we do not choose $q^*$, either $q^*$ is not incident to any addable edge or the maximal blocked addable thin edge $e_{q^*}$ identified for $q^*$ has value no greater than $\val(e_q)$. In the former case, $A_{q^*}$ has at most $\tau/\lambda$ worth of its resources not covered by ${\cal X}^{t_j}_{\leq j}\cup {\cal Y}^{t_j}_{\leq j-1}\cup {\cal I}^{t_j}$.  Since $\val(A_{q^*}) \geq \tau$ by Lemma~\ref{lem:induced-matching}, we conclude that more than $\tau - \tau/\lambda$ worth of thin resources in $A_{q^*}$ are already covered by ${\cal X}^{t_j}_{\leq j}\cup {\cal Y}^{t_j}_{\leq j-1}\cup {\cal I}^{t_j}$.  Note that $\tau-\tau/\lambda \geq \tau-w_q$ as $w_q = \val(e_q) \geq \tau/\lambda$ by invariant 1 in Table~\ref{tb:invar}.  In the latter case, it must be  that $\val(e_{q^*}) \leq \val(e_q) = w_{q} < \tau$.  Recall that a maximal blocked addable edge has value in $[\tau/\lambda,\tau+\tau/\lambda]$.   Therefore, given that  $\val(e_{q^*}) < \tau$ and yet $e_{q^*}$ is maximal, $e_{q^*}$ must contain all the thin resources desired by $q^*$ but not yet covered by ${\cal X}^{t_j}_{\leq j}\cup {\cal Y}^{t_j}_{\leq j-1}\cup {\cal I}^{t_j}$. Therefore, $A_{q^*}$ has at least $\tau - value(e_{q^*}) \geq \tau - w_{q}$ worth of thin resources covered by ${\cal X}^{t_j}_{\leq j}\cup {\cal Y}^{t_j}_{\leq j-1}\cup {\cal I}^{t_j}$. 
	\end{proof}
\end{quote}

The total value of resources shared by ${\cal X}^{t_j}_{\leq j}\cup {\cal Y}^{t_j}_{\leq j-1}$ and $A_{\varphi'(\dm_{\varphi'} \cap X'_j)}$ is at least the total value of those shared by ${\cal X}^{t_j}_{\leq j}\cup {\cal Y}^{t_j}_{\leq j-1} \cup {\cal I}^{t_j}$ and $A_{\varphi'(\dm_{\varphi'} \cap X'_j)}$ minus $\val({\cal I}^{t_j})$.  Since none of the layers below the $j$-th layer is collapsible, by Lemma~\ref{lem:non-collapse}(i), $|{\cal I}^{t_j}| \leq \mu|Y^{t_j}_{\leq j-1}|$.  By invariant~1 in Table~\ref{tb:invar}, edges in ${\cal I}^{t_j}$ have values at most $2\tau/\lambda$.  Therefore, by applying Claim~\ref{cl:lower-1} to all players in $\dm_{\varphi'} \cap X'_{j}$, the total value of resources shared by ${\cal X}^{t_j}_{\leq j}\cup {\cal Y}^{t_j}_{\leq j-1}$ and $A_{\varphi'(\dm_{\varphi'} \cap X'_{j})}$ is at least
\[
	\sum_{q \in \dm_{\varphi'}\cap X'_{j}}(\tau - w_q) - \val({\cal I}^{t_j}) 
	\geq \sum_{q \in \dm_{\varphi'}\cap X'_{j}}(\tau - w_q) 
	 	 -\frac{2\mu\tau}{\lambda}|Y^{t_j}_{\leq j-1}|.
\]
Comparing ${\cal X}^{t_j}_{\leq j}\cup {\cal Y}^{t_j}_{\leq j-1}$ and ${\cal X}'_{\leq j}\cup {\cal Y}'_{\leq j-1}$,  the only difference is ${\cal X}^{t_j}_{j} \setminus {\cal X}'_{j}$ because ${\cal X}^{t_j}_c = {\cal X}'_c$, ${\cal Y}^{t_j}_c = {\cal Y}'_c$ for all $c \in [1, j-1]$, and ${\cal X}'_j \subseteq {\cal X}^{t_j}_j$.  Since none of the layers below the $j$-th layer is collapsible, Lemmas~\ref{lem:invar} and~\ref{lem:non-collapse} imply that
	$|{\cal X}^{t_{j}}_{j}| - |{\cal X}'_{j}| 
	= z'_j - |{\cal X}'_{j}| 
	\leq \sum_{i=1}^{j} z'_i - |{\cal X}'_{\leq j}| \leq \mu|Y'_{\leq j-1}|$.
Also each blocked addable edge has value at most $\tau + \tau/\lambda$.  Therefore, 
	$\val({\cal X}^{t_j}_{j} \setminus {\cal X}'_{j}) 
	\leq (\tau+\tau/\lambda)\mu|Y'_{\leq j-1}|$.
As a result, the total value of resources shared by ${\cal X}'_{\leq j}\cup {\cal Y}'_{\leq j-1}$ and $A_{\varphi'(\dm_{\varphi'} \cap X'_{j})}$ is at least
\[
	\sum_{q \in \dm_{\varphi'}\cap X'_{j}}(\tau - w_q)
		 - \frac{2\mu\tau}{\lambda}|Y^{t_j}_{\leq j-1}|
		 - (\tau+\tau/\lambda)\mu|Y'_{\leq j-1}| 
	= \sum_{q \in \dm_{\varphi'}\cap X'_{j}}(\tau - w_q)
		 - (\tau+3\tau/\lambda)\mu|Y'_{\leq j-1}|.
\]
Recall that $X'_1 = \emptyset$ by the initialization of the building phase. Summing the right hand side of the above inequality over all $j \in [2,k+1]$, we conclude that the total value of resources shared by ${\cal X}'_{\leq k+1}\cup {\cal Y}'_{\leq k}$ and $A_{\im_{\varphi'}}$ is at least
\begin{equation}
	\sum_{q \in \dm_{\varphi'}}(\tau - w_q) 
	- (\tau+3\tau/\lambda)\mu\sum_{j = 2}^{k+1}|Y'_{\leq j-1}|.
\label{eq:lower-6}
\end{equation}
By the definition of $k$, for all $j \in [2,k]$, $|Y'_j| \geq \sqrt{\mu}|Y'_{\leq j-1}|$.  By invariant~4 in Table~\ref{tb:invar}, $\sum_{j=2}^{k} |Y'_j| \leq |Y'_{\leq k}|$.  Therefore,
\begin{equation}
\label{eq:lower-7}
	\sum_{j = 2}^{k+1}|Y'_{\leq j-1}| 
	\leq  \frac{1}{\sqrt{\mu}}\sum_{j=2}^k |Y'_{j}| + |Y'_{\leq k}| 
	\leq \left(\frac{1}{\sqrt{\mu}} + 1\right)|Y'_{\leq k}| 
	\leq \frac{2}{\sqrt{\mu}}|Y'_{\leq k}|.
\end{equation}
Recall that $B$ is the set of resources covered by edges in ${\cal X}'_{\leq k+1}\cup {\cal Y}'_k$. Substituting (\ref{eq:lower-7}) into (\ref{eq:lower-6}) gives
\begin{equation}
\label{eq:lower-key-2}
\val(B \cap A_{\im_{\varphi'}}) \geq \sum_{q \in \dm_{\varphi'}}(\tau - w_q) - (2\tau+6\tau/\lambda)\sqrt{\mu}|Y'_{\leq k}|.
\end{equation}
Combining \eqref{eq:lower-key-2} and \eqref{eq:lower-key-1} completes the proof:
\begin{align*}
	\val({\cal X}'_{\leq k+1}\cup {\cal Y}'_{\leq k}) 
	 \geq & 
	\val(B \cap A_{\overbar{P}_{M^*} \setminus A_{\im_{\varphi'}}}) + \val(B \cap A_{\im_{\varphi'}}) \\
	\geq & (\tau - \tau/\lambda)(|Y'_{\leq k}|- |\dm_{\varphi'}|) + \sum_{q \in \dm_{\varphi'}}(\tau - w_q) - \\
	 & (2\lambda\mu  + 2\lambda\sqrt{\mu} +6\sqrt{\mu})\frac{\tau}{\lambda}|Y'_{\leq k}|.
\end{align*}

%% file: src/conclusion.tex
\section{Conclusion}

We show that for any constant $\delta \in (0,1)$, a $(6+\delta)$-approximate solution can be computed for the restricted fair allocation problem in polynomial time.  There is still a gap between the current best estimation ratio and the approximation ratio $6+\delta$ achieved by this paper.  Whether the approximation ratio can match the estimation ratio remains an open problem.
 
The problem can be generalized slightly to the RAM model where the values of resources are non-integers.  Only the binary search scheme needs to adjusted.  When we work on an interval $[a,b]$ for $\tau$, we will recurse on either $[a,(a+b)/2]$ or $[(a+b)/2,b]$ depending on the outcome of solving the reduced problem.  We terminate the binary search when $b \leq (1+\delta)a$.  This gives an approximation ratio of $6 + \delta'$ for a constant $\delta' \in (0,1)$ depending on $\delta$.  